\documentclass{article}

\PassOptionsToPackage{sort&compress}{natbib}

\usepackage[preprint]{neurips_2026}

\makeatletter
\renewcommand{\@notice}{}
\makeatother

\usepackage[utf8]{inputenc} 
\usepackage[T1]{fontenc}    
\usepackage{hyperref}       
\usepackage{url}            
\usepackage{booktabs}       
\usepackage{siunitx}
\usepackage{amsmath}
\usepackage{amssymb}
\usepackage{mathtools}
\usepackage{amsthm}

\usepackage[capitalize,noabbrev]{cleveref}
\AddToHook{cmd/appendix/before}{%
    \crefalias{section}{appendix}%
    \crefalias{subsection}{appendix}
}
\usepackage{nicefrac}       
\usepackage[final]{microtype}      
\usepackage[dvipsnames, x11names]{xcolor}         
\usepackage[createShortEnv]{proof-at-the-end}
\usepackage{placeins}
\title{Epsilon-Nash Equilibria in History-Dependent SA-MDPs}

\author{%
  Brandon Gary Kaplowitz\textsuperscript{1,*}\\
\And
  Dominik Bohnet Zurcher\textsuperscript{2,*} \AND
  Akash Agrawal\textsuperscript{2,3} \And
  Tala Jafari\textsuperscript{2} \And
  Christian Schroeder de Witt\textsuperscript{1} \And
  Paul W. Goldberg\textsuperscript{2} \\[0.5em]
  \textsuperscript{1}Department of Engineering Science, University of Oxford \\
  \textsuperscript{2}Department of Computer Science, University of Oxford \\
  \textsuperscript{3}ML Alignment and Theory Scholars \\
  \textsuperscript{*}Equal contribution. Correspondence to \texttt{brandon.kaplowitz@eng.ox.ac.uk}
  .
}
\usepackage{algorithm}
\usepackage{algorithmic}
\usepackage{tikz}
\usepackage{subcaption}
\usetikzlibrary{automata, positioning, arrows.meta,fit,backgrounds,calc}
\usetikzlibrary{positioning, shapes.geometric}
\usepackage{changepage}
\newcommand{\Prob}{\mathbb{P}}
\theoremstyle{plain}
\newtheorem{theorem}{Theorem}[section]

\newtheorem{lemma}[theorem]{Lemma}

\theoremstyle{definition}
\newtheorem{definition}[theorem]{Definition}

\theoremstyle{remark}

\usepackage[textsize=tiny]{todonotes}

\begin{document}

\maketitle

\begin{abstract}
We study state-adversarial Markov decision processes (SA-MDPs) as games of observation-space attacks: at each step, an agent selects an action from a received observation while an adversary---who knows the true state the agent is in---chooses a perturbed observation within a state-dependent proximity set. While existing work focuses on Markovian policies, we develop a solution concept and computational approach for SA-MDPs under history dependence. History dependence can materially change equilibrium outcomes and can force both the agent and the adversary to adapt their strategies. First, we prove the non-existence of universal (agnostic of the initial state distribution) history-dependent equilibrium policies. Our main result presents the first algorithmic route to computing $\epsilon$-approximations of initial-state dependent equilibria. We do so by reducing SA-MDPs to a strategically equivalent constrained zero-sum one-sided partially observable stochastic game. We test our algorithm on small analytically verifiable games and show that it scales to larger, more realistic benchmarks, including Atari Freeway rollouts with a 12-period ahead horizon.





\end{abstract}
\section{Introduction} 
Reinforcement learning under adversarially manipulated observations poses fundamental challenges that are not captured by standard Markovian robustness models. In this work, we study state-adversarial Markov decision processes with history-dependent policies, characterize their equilibrium structure, and provide the first algorithmic approach for computing approximate history-dependent equilibria.

We consider a state-adversarial Markov decision process (SA-MDP) modeling an observation-space attack game. In each time step $t$, an \emph{adversary} (player 2) first selects an observation $o^t$ to send to the agent, and an \emph{agent} (player 1) then takes an action $a^t_1$ based on its observations. The adversary knows the true state $s^t$ and must choose $o^t$ from an allowed set $B(s^t)$ (a \emph{proximity constraint} neighborhood of $s^t$), which limits how far the observation can deviate from reality. The agent receives $o^t$ (but not $s^t$) and the environment then transitions to a new state $s^{t+1}$ with probability $p(s^{t+1}\mid s^t,a^t_1)$, and a reward $R(s^t,a^t_1)$ is obtained by the agent (the adversary’s payoff is the negation). This repeats indefinitely with discount factor $\gamma\in(0,1)$. The adversary’s goal is to minimize the agent’s total expected reward by cleverly manipulating observations. Such SA-MDP models have been studied in robust RL and planning settings, usually under the restriction that both players use \emph{Markovian} (memoryless) policies.

Most prior work on SA-MDPs assumes stationary optimal policies depending only on the current state (for the adversary) or observation (for the agent). Under this assumption, the interaction can be formulated as a two-player zero-sum game with a stationary equilibrium policy pair $(\pi_1^*,\pi_2^*)$ (for agent and adversary respectively). However, as we argue below, an agent that can condition its strategy on past observations (i.e., a \emph{history-dependent} policy) can potentially achieve higher rewards than any memoryless policy by detecting and exploiting patterns in the observation sequence. We therefore seek a solution concept allowing \textbf{history-dependent policies} in SA-MDPs, and develop the first algorithm to compute such an equilibrium.

\subsection{Related Work}
This paper is inspired by the foundational work of cooperative multi-agent learning, aiming to build adversarially-robust, collaborative autonomous agents \citep{panait2005cooperative, wang2017cooperative, cai2022cooperative, wang2022cooperative}. Early work studied adversarial attacks on reinforcement-learning policies mainly as empirical attack constructions and vulnerability demonstrations: adversarial perturbations to observations can substantially degrade learned policies, and attack timing can matter over trajectories \citep{huang2017adversarial, lin2017tactics}; policy-induction attacks during learning were also considered early on \citep{behzadan2017vulnerability}. The SA-MDP line then formalized worst-case attacks on state observations and robust learning in this setting \citep{zhang2020robust}. Follow-up work studied stronger attacks and robust training in closely related fixed-agent, alternating-training, or restricted-policy-class settings, including learned optimal adversaries with alternating training \citep{zhang2021robust}, efficient strongest-attack methods \citep{sun2021strongest}, and robust/adversarial-training methods such as those of \citet{oikarinen2021robust} and \citet{liang2022efficient}.

Several nearby game-theoretic extensions are also relevant \citep{pendharkar2012game}. Illusory attacks impose information-theoretic detectability constraints on observation attacks \citep{franzmeyer2023illusory}. Broader online attack/defense formulations consider manipulation of states, observations, actions, and rewards \citep{mcmahan2023optimal, sarkadi2019modelling, motwani2024secret, wang2022novel, tu2021adversarial, BirmpasGHCRV21}. Temporally coupled perturbations have been modeled as a partially observable two-player zero-sum game and solved approximately via equilibrium computation \citep{liang2023game, khakpour2022partially}. In multi-agent settings, state-adversarial Markov games can fail to admit standard optimal-policy or robust-Nash solution concepts \citep{han2022solution}.

Our algorithmic approach also builds directly on the planning and game-solving literature: partially observable Markov decision processes (POMDPs) provide the underlying partial-observability formalism on the agent side \citep{kaelbling1998planning}, Heuristic search value iteration (HSVI) supplies the core search/value-iteration template \citep{smith2012heuristic}, and \citet{horak2023solving} gives the immediate zero-sum one-sided POSG framework and solver that our reduction and algorithm build on. Relative to this literature, our contribution is to study observation-space SA-MDPs with history-dependent policies, prove that universal history-dependent state-robust policies need not exist, and give a reduction-based route to computing $\epsilon$-approximate equilibria for a given initial state distribution.

\subsection{Background and Notation}

For any set $X$, let $\Delta(X)$ denote the set of probability distributions over $X$. Notation of variables is that $x^t_i$ indicates the ith player, ith entry, or ith index of variable $x$ at time or game-stage $t$. Any absent notation is to be understand as the vector of the variable for all $i$ or $t$ in the relevant set of potential outcomes. This split notation is to prevent confusion between a time subscript of $2t$ and ``$2,t$", which both appear in the paper.

\citet{zhang2020robust} introduce the state-adversarial Markov decision process (SA-MDP) to model an RL agent (player 1) interacting with an environment under adversarial observation perturbations (player 2). An SA-MDP can be defined as a tuple $M = (S, A_1, R, p, \gamma, B)$, where $S$ is the true state space, $A_1$ the agent's action space, $R: S \times A_1 \to \mathbb{R}$ the reward function, $p: S \times A_1 \to \Delta(S)$ the state transition probabilities, $\gamma\in(0,1)$ the discount factor, and $B$ the function that specifies the allowable perturbations of the adversary. In particular, for each true state $s\in S$, $B(s)\subseteq S$ is the set of states that the \emph{adversary} can present to our \emph{agent} as a perturbed observation. We denote the adversary’s (possibly stochastic) policy by $\pi_2: S \to \Delta(S)$, where $\pi_2(s)\in \Delta(B(s))$. At each time step, the agent observes $\tilde{s} \sim \pi_2(s)$, and then selects an action $a\in A_1$ according to its policy $\pi_1: S \to \Delta(A_1)$ (which is applied to the observation $\tilde{s}$). The environment then transitions to the next true state $s'$ drawn from $p(s'|s, a)$, and a reward $R(s,a)$ is emitted. Notably, the adversary influences the agent's perception but \emph{does not alter} the underlying state transition; $s'$ depends on the actual $s$ and $a$ taken, not directly on $\tilde{s}$. The adversary’s objective is assumed to be worst-case in that it aims to minimize the agent’s total return (e.g., by choosing perturbations that lead the agent to make suboptimal decisions). This worst-case observation model cannot be captured by a standard partially-observable MDP, since a POMDP uses fixed observation probabilities \citep{kaelbling1998planning}, whereas here the observation distribution observed by the agent depends on the adversary's policy.

Formally, we can define the value of a given stationary agent policy $\pi_1$ under a specific adversary $\pi_2$ starting from state $s$ as
$$ V(s;\pi_1,\pi_2) = \mathbb{E}_{\pi_1,\pi_2}\Big[ \sum_{t=0}^\infty \gamma^t R(s^t, a^t_1) \,\Big|\, s^{0} = s \Big] ,$$
where $a^t_1 \sim \pi_1(\cdot \mid \tilde{s}^t)$, $\tilde{s}^t \sim \pi_2(s^t)$, and $s^{t+1}\sim p(\cdot \mid s^t, a^t_1)$. Of particular interest is the optimal adversary for a given $\pi_1$, $\pi^*_{2,\pi_1}$. We define the optimal adversarial value function as:
$$V(s;\pi_1,\pi^*_{2,\pi_1}) =  \min_{\pi_2} V(s;\pi_1,\pi_2).$$
In an SA-MDP, the agent faces a two-player zero-sum game against the adversary, who chooses the worst observable state at each step. \citet{zhang2020robust} show that for a fixed $\pi_1$, the value function under the optimal adversary $\pi^*_{2,\pi_1}$ satisfies a robust Bellman equation. Specifically, letting $\underline V(s;\pi_1) = V(s;\pi_1,\pi^*_{2,\pi_1})$ denote the value under the strongest adversary, we have:
\begin{equation}
\label{eq:robust-bellman}
\begin{aligned}
\underline V(s;\pi_1)
= &\min_{\tilde{s}\in B(s)}
\sum_{a\in A_1} \pi_1(a \mid \tilde{s}) \ \cdot \sum_{s'\in S} p(s' \mid s, a)
\Big[ R(s,a) + \gamma \underline V(s';\pi_1) \Big].
\end{aligned}
\end{equation}
which defines a contraction operator on value functions. Intuitively, at state $s$ the adversary chooses a perturbation $\tilde{s}\in B(s)$ that minimizes the expected return and expected distribution over future values, knowing the agent will act according to $\pi_1$ on $\tilde{s}$. This worst-case Bellman operator allows ``policy evaluation'' under the optimal adversary.

However, in order to consider history dependent policies, we introduce another game, the Zero-Sum One-Sided Partially Observable Stochastic Game, formalized by  \cite{horak2023solving}. 

\begin{adjustwidth}{0em}{0pt}
\begin{definition}[Zero-Sum One-Sided Partially Observable Stochastic Game]
    A Zero-Sum One-Sided Partially Observable Stochastic Game (or zero-sum OS-POSG) is  a tuple $G = (S,A_1, A_2, O, T, R, \gamma)$ where: 
    \begin{itemize}
        \item S is a finite set of game states,
        \item $A_1$ and $A_2$ are finite sets of actions of player 1 and player 2, respectively,
        \item O is a finite set of observations 
        \item $T(\cdot| s, a_1, a_2) \in \Delta(O \times S)$ represents probabilistic transition function for every $(s, a_1, a_2) \in S \times A_1 \times A_2$,
        \item $R : S \times A_1 \times A_2 \rightarrow \mathbb{R}$ is player 1's reward function,
        \item $\gamma \in (0, 1)$ is a discount factor.
    \end{itemize} 
\end{definition}
\end{adjustwidth}

\noindent The play unfolds as follows.  First, an initial state 
$s^{0} \sim b^0$ is drawn from the probability distribution \(b^0\). Then, for each stage, the current state \(s^{t}\) is revealed to player 2 but remains hidden from player 1. Simultaneously, player 1 chooses \(a^{t}_{1} \in A_1\) and player 2 chooses \(a^{t}_{2} \in A_2\). An unobserved reward \(R\bigl(s^{t},a^{t}_{1},a^{t}_{2}\bigr)\) is awarded to player 1, while player 2's payoff is \(-R\bigl(s^{t},a^{t}_{1},a^{t}_{2}\bigr)\). The system transitions to a new state \(s^{t+1}\) and emits an observation \(o^{t}\) according to $T\bigl(o^{t},s^{t+1} \mid s^{t},a^{t}_{1},a^{t}_{2}\bigr).$ Player 2 then observes the full outcome of the stage: $\bigl(s^{t},\,a^{t}_{1},\,a^{t}_{2},\,o^{t}\bigr).$ Player 1, by contrast, sees only $\bigl(a^{t}_{1},\,o^{t}\bigr),$ without learning \(a^{t}_{2}\), \(s^{t}\), or \(s^{t+1}\).

\section{Theoretical results}

In Section~\ref{sec:nonexist}, we show the non-existence of universal Nash equilibria, i.e. policies independent of the initial state distribution. Assuming an initial state distribution is given, in Section~\ref{sec:history}, we show the importance of history-dependent policies, arguing that agents are incentivized to use them, and adversaries must then modify their policies accordingly. We then show in Section~\ref{sec:modified} how to reduce the SA-MDP problem to a modified version of the zero-sum OS-POSG, which allows us to find an epsilon-approximate Nash equilibrium for the original SA-MDP. This equilibrium allows us to find policies that guarantee near optimal rewards.

\subsection{Non-Existence of Universal History Dependent Nash Equilibrium Policies}\label{sec:nonexist}

In a zero-sum game, the agent's optimal strategy can be defined in two equivalent ways: as a robust policy that maximizes value against the worst-case adversary, or as part of a Nash equilibrium (NE) of the game. 
We show that these notions coincide for SA-MDPs. 
\begin{lemmaE}[Robust Optimality $\iff$ Nash Equilibrium][restate,text link section]
\label{lemma:robust-nash}
Consider a zero-sum SA-MDP with finite state, action, and observation sets, nonempty perturbation sets $B(s)$, bounded rewards, a fixed initial distribution $b\in\Delta(S)$, and unrestricted full-history behavioral strategy spaces $\Pi_1$ and $\Pi_2$ for the agent and the adversary, respectively. We assume perfect recall in which each player's private history records all previously observed information and all of that player's previous actions.
For a strategy profile $(\pi_1,\pi_2)\in\Pi_1\times\Pi_2$, let $V(b;\pi_1,\pi_2)$ denote the agent's expected return for the horizon under consideration when $s^{0}\sim b$; when the horizon is $N<\infty$, we also write $V^N(b;\pi_1,\pi_2)$ to make the horizon explicit.
For a finite horizon, we allow $\gamma\in(0,1]$, with $\gamma=1$ corresponding to the undiscounted return; for an infinite discounted horizon, we require $\gamma\in(0,1)$. In either case,
the following equality holds for every initial distribution $b\in\Delta(S)$:
\[
\begin{aligned}
&\arg\max_{\pi_1\in\Pi_1}\min_{\pi_2\in\Pi_2}V(b;\pi_1,\pi_2)\\
&\quad =
\left\{\pi_1\in\Pi_1:\ \exists\pi_2\in\Pi_2\text{ with }(\pi_1,\pi_2)
\text{ a Nash equilibrium when }s^{0}\sim b\right\}.
\end{aligned}
\]
Both sets are nonempty.
\end{lemmaE}
\begin{proofE}
Throughout, player \(1\) is the maximizing player, and player \(2\) is the minimizing player.

\medskip\noindent\emph{Finite horizon.}
Fix a horizon \(N<\infty\). Since all underlying alphabets are finite, the induced \(N\)-stage extensive-form zero-sum game has finitely many histories and information sets. By assumption, it has perfect recall.

Let \(P_i^N\) be the finite set of player \(i\)'s pure complete contingent plans, where a plan specifies an action at every information set of that player, including information sets that the plan itself prevents from being reached. Let
\[
M_i^N\coloneqq\Delta(P_i^N)
\]
be the corresponding mixed-strategy simplex. If \(u_N(\mu,\nu)\) denotes the expected payoff under \((\mu,\nu)\in M_1^N\times M_2^N\), then \(u_N\) is bilinear. Von Neumann's minimax theorem, therefore, gives
\((\mu_N^*,\nu_N^*)\in M_1^N\times M_2^N\) such that
\[
u_N(\mu,\nu_N^*)
\le
u_N(\mu_N^*,\nu_N^*)
\le
u_N(\mu_N^*,\nu)
\qquad
\forall(\mu,\nu)\in M_1^N\times M_2^N.
\]
Therefore, by Kuhn's realization-equivalence theorem for finite perfect-recall games, there are behavioral strategies \(\pi_{1,N}^*\) and \(\pi_{2,N}^*\) that are realization-equivalent to \(\mu_N^*\) and \(\nu_N^*\), respectively \citep{kuhn1953extensive}.
Here, realization equivalence means equality of the induced distribution over terminal histories against every strategy of the opponent. 
Conversely, every behavioral strategy in this finite game has a realization-equivalent mixed strategy over complete contingent plans.

Consequently, for arbitrary behavioral deviations \(\pi_1\) and \(\pi_2\), choose realization-equivalent mixed strategies \(\mu_{\pi_1}\) and \(\nu_{\pi_2}\).
The preceding normal-form saddle inequalities imply
\[
V^N(b;\pi_1,\pi_{2,N}^*)
=
u_N(\mu_{\pi_1},\nu_N^*)
\le
u_N(\mu_N^*,\nu_N^*)
=
V^N(b;\pi_{1,N}^*,\pi_{2,N}^*)
\]
and
\[
V^N(b;\pi_{1,N}^*,\pi_{2,N}^*)
=
u_N(\mu_N^*,\nu_N^*)
\le
u_N(\mu_N^*,\nu_{\pi_2})
=
V^N(b;\pi_{1,N}^*,\pi_2).
\]
Thus \((\pi_{1,N}^*,\pi_{2,N}^*)\) is a behavioral saddle point of the \(N\)-stage game.

\medskip\noindent\emph{Infinite discounted horizon.}
We obtain an infinite-horizon saddle point as a limit of finite-horizon saddle points.

For \(n\geq 0\), let \(\mathcal H_i^n\) denote the finite set of player \(i\)'s private histories immediately before the stage-\(n\) action. Under the information structure above, generic elements of \(\mathcal H_1^n\) and \(\mathcal H_2^n\), respectively, have the forms
\[
h^{n}_{1}
=
\bigl(o^{0},a^{0}_{1},\ldots,o^{n-1},a^{n-1}_{1},o^{n}\bigr)
\]
and
\[
h^{n}_{2}
=
\bigl(s^{0},a^{0}_{2},o^{0},a^{0}_{1},s^{1},\ldots,a^{n-1}_{2},o^{n-1},a^{n-1}_{1},s^{n}\bigr),
\]
with \(h^{0}_{1}=(o^{0})\) and \(h^{0}_{2}=(s^{0})\), where \(a^t_2=o^t\in B(s^t)\) is the adversary's stage-\(t\) action. Let
\[
\mathcal H_i\coloneqq\bigcup_{n\geq 0}\mathcal H_i^n.
\]
Because all underlying alphabets are finite, each \(\mathcal H_i^n\) is finite, so \(\mathcal H_i\) is countable.
For \(h\in\mathcal H_i\), let \(A_i(h)\) be the nonempty finite set of actions available to player \(i\) at \(h\). Thus, \(A_1(h)=A_1\) for the agent, while \(A_2(h)=B(s)\) at an adversary history whose current true state is \(s\). Define
\[
\Pi_i
\coloneqq
\prod_{h\in\mathcal H_i}\Delta(A_i(h))
\]
and endow it with the product topology.

Each factor \(\Delta(A_i(h))\) is a compact metrizable simplex in a finite-dimensional Euclidean space. Tychonoff's theorem therefore implies that \(\Pi_i\) is compact in the product topology \citep[Theorem~37.3]{munkres2000topology}. Moreover, the probability coordinates \((h,a)\), where \(h\in\mathcal H_i\) and \(a\in A_i(h)\), form a countable set. Identifying each behavioral choice with its action probabilities therefore realizes \(\Pi_i\), with the product topology above, as a subspace of a countable product of copies of \([0,1]\), which is metrizable in the product topology \citep[Theorem~20.5]{munkres2000topology}. Thus, each \(\Pi_i\), and hence the finite product \(\Pi_1\times\Pi_2\), is compact and metrizable. It is therefore sequentially compact \citep[Theorem~28.2]{munkres2000topology}.

For \(N\ge1\), define the truncated discounted payoff
\[
V^N(b;\pi_1,\pi_2)
\coloneqq
\mathbb E_{s^{0}\sim b,\pi_1,\pi_2}
\left[
\sum_{t=0}^{N-1}\gamma^tR_t
\right].
\]
For each player \(i\), \(V^N(b;\cdot,\cdot)\) depends only on the behavioral coordinates indexed by the finite set \(\bigcup_{n=0}^{N-1}\mathcal H_i^n\), namely the private histories immediately preceding actions in the first \(N\) stages.
Expanding the expectation as a sum over the finitely many length-\(N\) paths shows that it is a polynomial in those coordinates and therefore is continuous in the product topology.

Choose \(\bar R<\infty\) such that \(|R_t|\le\bar R\) almost surely. Then, uniformly over \((\pi_1,\pi_2)\in\Pi_1\times\Pi_2\),
\begin{equation}
\label{eq:discounted-tail-uniform}
\left|
V(b;\pi_1,\pi_2)-V^N(b;\pi_1,\pi_2)
\right|
\le
\sum_{t=N}^{\infty}\gamma^t\bar R
=
\frac{\bar R\gamma^N}{1-\gamma}
\eqqcolon\varepsilon_N
\longrightarrow0.
\end{equation}
Thus \(V(b;\cdot,\cdot)\) is the uniform limit of the continuous functions \(V^N(b;\cdot,\cdot)\), so
\(V(b;\cdot,\cdot)\) is jointly continuous on \(\Pi_1\times\Pi_2\).

For every \(N\), let \((\pi_1^N,\pi_2^N)\) be a behavioral saddle point of the \(N\)-stage discounted game, whose existence follows from the finite-horizon argument, and extend these strategies arbitrarily at histories occurring after stage \(N\).
Sequential compactness of \(\Pi_1\times\Pi_2\) gives a subsequence \(N_k\to\infty\) and a profile
\((\pi_1^*,\pi_2^*)\) such that
\[
(\pi_1^{N_k},\pi_2^{N_k})
\longrightarrow
(\pi_1^*,\pi_2^*).
\]

Fix arbitrary infinite-horizon strategies
\((\pi_1,\pi_2)\in\Pi_1\times\Pi_2\).
Their restrictions to the first \(N_k\) stages are admissible behavioral deviations in the \(N_k\) 
stage game. 
Therefore
\[
V^{N_k}(b;\pi_1,\pi_2^{N_k})
\le
V^{N_k}(b;\pi_1^{N_k},\pi_2^{N_k})
\le
V^{N_k}(b;\pi_1^{N_k},\pi_2).
\]
Using \eqref{eq:discounted-tail-uniform} gives
\[
V(b;\pi_1,\pi_2^{N_k})
\le
V(b;\pi_1^{N_k},\pi_2^{N_k})+2\varepsilon_{N_k}
\]
and
\[
V(b;\pi_1^{N_k},\pi_2^{N_k})
\le
V(b;\pi_1^{N_k},\pi_2)+2\varepsilon_{N_k}.
\]
Joint continuity of \(V(b;\cdot,\cdot)\), convergence of the strategy profiles, and \(\varepsilon_{N_k}\to0\) now imply
\[
V(b;\pi_1,\pi_2^*)
\le
V(b;\pi_1^*,\pi_2^*)
\le
V(b;\pi_1^*,\pi_2)
\qquad
\forall(\pi_1,\pi_2)\in\Pi_1\times\Pi_2.
\]
Hence, \((\pi_1^*,\pi_2^*)\) is an infinite-horizon behavioral saddle point.

\medskip\noindent\emph{Equality of the two policy sets.}
In either horizon model, let \((\pi_1^0,\pi_2^0)\) be a saddle point and set
\[
v\coloneqq V(b;\pi_1^0,\pi_2^0).
\]
The saddle inequalities give
\[
v
=
\max_{\pi_1\in\Pi_1}\min_{\pi_2\in\Pi_2}V(b;\pi_1,\pi_2)
=
\min_{\pi_2\in\Pi_2}\max_{\pi_1\in\Pi_1}V(b;\pi_1,\pi_2).
\]

Suppose that \(\widehat\pi_1\) is robust-optimal. Then
\[
\min_{\pi_2}V(b;\widehat\pi_1,\pi_2)=v,
\]
and hence
\[
V(b;\widehat\pi_1,\pi_2)\ge v
\qquad\forall\pi_2\in\Pi_2.
\]
The saddle property of \((\pi_1^0,\pi_2^0)\) also gives
\[
V(b;\pi_1,\pi_2^0)\le v
\qquad\forall\pi_1\in\Pi_1.
\]
In particular, \(V(b;\widehat\pi_1,\pi_2^0)=v\), and consequently
\[
V(b;\pi_1,\pi_2^0)
\le
V(b;\widehat\pi_1,\pi_2^0)
\le
V(b;\widehat\pi_1,\pi_2)
\qquad\forall(\pi_1,\pi_2).
\]
Thus \((\widehat\pi_1,\pi_2^0)\) is a saddle point, equivalently a Nash equilibrium.

Conversely, if \((\bar\pi_1,\bar\pi_2)\) is any Nash equilibrium, then, because the game is zero-sum, it is a saddle point, and its payoff equals the game value \(v\).
Hence
\[
\min_{\pi_2\in\Pi_2}V(b;\bar\pi_1,\pi_2)
=
V(b;\bar\pi_1,\bar\pi_2)
=
v,
\]
so \(\bar\pi_1\) is robust-optimal. The existence of a saddle point also proves that both sets in the statement are nonempty.
\end{proofE}

We use this lemma to demonstrate the impossibility of finding an algorithm that always returns policies Nash equilibrium agnostic of the initial state distribution.  \citet{zhang2020robust} showed that there exist SA-MDPs where there is no Markovian optimal agent policy that is independent of the initial state distribution $b$. They leave open the question of whether the same holds for history dependent policies.  We revisit their concrete example and show that no history dependent policy can achieve the minimax value against all adversarial responses. According to the previous lemma, this implies that no initial-state-distribution-agnostic Nash equilibrium exists. 

\begin{theoremE}[][restate,text link section]

\label{thm:nonexist}
There exists a SA-MDP such that no history-dependent policy is state-robust optimal.
\end{theoremE}


\begin{proofE}
\tikzset{
  state/.style={
    circle,
    draw,
    thick,
    minimum size=1.2cm,
    inner sep=0pt,
    font=\small,
    align=center
  }
}
\begin{figure}
    \centering
    \caption{MDP without History-Dependent State-Robust Policy}
    \label{fig:no-robust-opt}

    \begin{tikzpicture}[scale=0.67, transform shape, node distance=2.5cm, >=Latex]
      \node[state]            (S1)                {S$_1$};
      \node[state, below right=of S1] (S2)         {S$_2$};
      \node[state, below left=of S1] (S3)         {S$_3$};

      \draw[->] (S1) edge[loop above] node[above]
        {$a_1: R = 0$} (S1);
      \draw[->] (S2) edge[loop right] node[right]
        {$a_1: R = 1$} (S2);
      \draw[->] (S3) edge[loop left] node[left]
        {$a_1: R = 1$} (S3);

      \draw[->]
        (S1) to node[midway, above right, align=left]
          {$a_2: R=1$} (S2);

      \draw[->]
        (S2) to node[midway, below]
          {$a_2: R = 0$} (S3.east);

      \draw[->]
        (S3) to node[midway, left, align=right]
          {$a_2: R = 0$} (S1);
    \end{tikzpicture}
\end{figure}

We begin with a high-level overview of the proof and then proceed through the details.

Let $\pi_1^A$ be the policy that always plays $a_1$. In the SA-MDP of \cref{fig:no-robust-opt}, $\pi_1^A$ achieves robust value $0$ at $s_1$ and $\tfrac{1}{1-\gamma}$ at $s_2,s_3$, while uniformly mixing $a_1,a_2$ achieves $\tfrac{1}{2(1-\gamma)}$ everywhere. Any state-robust optimal $\pi_1$ must therefore match $\pi_1^A$ at $s_2,s_3$ and strictly improve on it at $s_1$. To improve at $s_1$ against the identity adversary, $\pi_1$ must, after some first observed (i.e. reported by adversary) history $h^\ast = (s_1,a_1,\ldots,s_1)$, place positive probability on $a_2$. But the agent conditions only on reported histories: an adversary starting from true state $s_3$ can fabricate the same $h^\ast$ by reporting $s_1$ at every step, since $\pi_1$ plays $a_1$ surely along the prefix and the true state remains $s_3$. Triggering $a_2$ at true $s_3$ strictly reduces the return, so $\underline V(s_3;\pi_1) < \tfrac{1}{1-\gamma}$, contradicting state-robust optimality at $s_3$.
For a history-dependent agent policy $\pi_1$ and state $s \in S$, define the worst-case value
\[
\underline V(s;\pi_1) \;:=\; \min_{\pi_2} V(s;\pi_1,\pi_2).
\]
A policy $\pi_1$ is \emph{state-robust optimal} if $\underline V(s;\pi_1) \geq \underline V(s;\pi_1')$ for every state $s$ and every alternative policy $\pi_1'$.

Consider the SA-MDP in \cref{fig:no-robust-opt}. Let $\pi_1^A$ denote the policy that always plays $a_1$, regardless of the reported history. Direct computation gives
\[
\underline V(s_1;\pi_1^A) = 0, \qquad \underline V(s_2;\pi_1^A) = \underline V(s_3;\pi_1^A) = \frac{1}{1-\gamma}.
\]
For $\gamma = 0.99$, these values are $0$, $100$, and $100$, respectively.

Now let $\pi_1^B$ be the policy that plays each of $\{a_1, a_2\}$ with probability $\tfrac{1}{2}$ at every history. Under $\pi_1^B$, the expected one-step reward is $\tfrac{1}{2}$ in every state, so
\[
\underline V(s;\pi_1^B) \;=\; \frac{1}{2(1-\gamma)} \quad \text{for all } s \in \{s_1, s_2, s_3\},
\]
which equals $50$ when $\gamma = 0.99$.

Therefore, any state-robust optimal policy $\pi_1$ must satisfy
\[
\underline V(s_1;\pi_1) \geq 50, \qquad \underline V(s_2;\pi_1) \geq 100, \qquad \underline V(s_3;\pi_1) \geq 100.
\]
Since no policy can achieve more than $\tfrac{1}{1-\gamma} = 100$ from $s_2$ or $s_3$, any state-robust optimal policy must in fact satisfy
\begin{equation}
\label{eq:robust-opt-conditions}
\underline V(s_2;\pi_1) = \underline V(s_3;\pi_1) = \frac{1}{1-\gamma} \quad \text{and} \quad \underline V(s_1;\pi_1) \geq \frac{1}{2(1-\gamma)} > 0.
\end{equation}

\medskip

\noindent\textbf{Contradiction argument.} Assume for contradiction that there exists a history-dependent policy $\pi_1$ satisfying \eqref{eq:robust-opt-conditions}.

Let $\pi_2^I$ denote the \emph{identity adversary} that always reports the true state. Since $\underline V(s_1;\pi_1) > 0$, we have in particular
\[
V(s_1;\pi_1,\pi_2^I) \;\geq\; \underline V(s_1;\pi_1) \;>\; 0.
\]
Hence, starting from true state $s_1$ against $\pi_2^I$, there must exist some first time step $n$ at which $\pi_1$ assigns positive probability to action $a_2$. Otherwise, $\pi_1$ would play $a_1$ surely at every step in $s_1$, yielding value $0$, which contradicts $V(s_1;\pi_1,\pi_2^I) > 0$.

By minimality of $n$, along the identity-adversary play starting from $s_1$, the reported history up to time $n$ must be
\[
h^* \;=\; (s_1, a_1, s_1, a_1, \ldots, s_1, a_1, s_1),
\]
i.e., $n$ repeated reports of $s_1$ with action $a_1$ taken with probability $1$ at all earlier time steps. After this history, the policy plays $a_2$ with positive probability:
\[
\Pr_{\pi_1}\!\bigl[a_2 \,\big|\, h^*\bigr] \;>\; 0.
\]

\medskip

\noindent\textbf{Mimicking the history from $s_3$.} Now consider instead starting from true state $s_3$, and let $\pi_2'$ be the adversary that reports $s_1$ for the first $n$ observations (recall $s_1 \in B(s_3)$, so this is feasible). Since $\pi_1$ conditions only on the \emph{reported} history, it plays $a_1$ with probability $1$ on the first $n-1$ steps---exactly as it does along $h^*$ from true state $s_1$. Because $s_3$ is absorbing under action $a_1$, the true state remains $s_3$ throughout these steps.

At time $n$, the reported history is again $h^*$, so the policy plays $a_2$ with positive probability. However, at true state $s_3$, action $a_2$ yields reward $0$ and transitions out of $s_3$, while $a_1$ yields reward $1$ and keeps the agent in $s_3$. Conditioning on the event $\{a_2 \text{ played at step } n\}$, which occurs with positive probability under $\pi_1$ given $h^*$, the discounted return is strictly less than $\tfrac{1}{1-\gamma}$. Therefore
\[
V(s_3;\pi_1,\pi_2') \;<\; \frac{1}{1-\gamma},
\]
which implies
\[
\underline V(s_3;\pi_1) \;\leq\; V(s_3;\pi_1,\pi_2') \;<\; \frac{1}{1-\gamma}.
\]

\medskip

This contradicts \eqref{eq:robust-opt-conditions}, which required $\underline V(s_3;\pi_1) = \tfrac{1}{1-\gamma}$. Hence no such history-dependent policy $\pi_1$ can exist, and no history-dependent policy is state-robust optimal in this SA-MDP.
\end{proofE}

\subsection{Why History Matters}\label{sec:history}

Most models assume Markovian (memoryless) policies for both the agent and the adversary, meaning both the adversary’s and agent's actions depend only on the current state or observation and not on past history. This yields a tractable two-player zero-sum framework with stationary, optimal policies. However, a purely Markovian agent might miss opportunities that arise from factoring in the history of play. Consider \autoref{fig:historymatters} where the transition probabilities do not depend on the actions, the rewards are inverted for state $s_1$ and $s_2$, state $s_3$ deterministically transitions to state $s_1$ without emitting a reward, and $B(s_1) = B(s_2) = \{s_1, s_2\}, \; B(s_3) = \{s_3\}$.

\tikzset{
  state/.style={
    circle,
    draw,
    thick,
    minimum size=1.2cm,
    inner sep=0pt,
    font=\small,
    align=center
  }
}

\begin{figure}
    \centering

    \begin{subfigure}[t]{0.49\linewidth}
        \centering
        \caption{History Matters}
        \label{fig:historymatters}

        \begin{tikzpicture}[scale=0.60, transform shape, node distance=3cm, >=Latex]
          \node[state, label=above:{\shortstack[l]{$a_1: R = 1$\\$a_2 : R = -1$}}]
                (S1) {S$_1$};

          \node[state, label=above:{\shortstack[l]{$a_1: R = -1$\\$a_2 : R = 1$}},
                right=5cm of S1] (S2) {S$_2$};

          \node[state, label=below:{$a_1, a_2: R = 0$},
                right=2cm of S1, yshift=-2.5cm] (S3) {S$_3$};

          \draw[->] (S1) edge[loop left] node[align=left] {$1/2$} (S1);
          \draw[->] (S2) edge[loop right] node[align=left] {$2/5$} (S2);

          \draw[->, bend left]
            (S1) to node[align=left, yshift=0.5cm] {$1/2$} (S2);

          \draw[->, bend left]
            (S2) to node[align=left, yshift=0.5cm] {$2/5$} (S1);

          \draw[->, bend left]
            (S2) to node[align=left, yshift=-0.5cm] {$1/5$} (S3);

          \draw[->, bend left]
            (S3) to node[align=left, yshift=-0.5cm] {$1$} (S1);
        \end{tikzpicture}
    \end{subfigure}
    \hfill
    \begin{subfigure}[t]{0.49\linewidth}
        \centering
        \caption{History Matters More}
        \label{fig:historymattersmore}

        \begin{tikzpicture}[scale=0.60, transform shape, node distance=3cm, >=Latex]
          \node[state, label=above:{\shortstack[l]{$a_1: R = 1$\\$a_2 : R = -1$}}]
                (S1) {S$_1$};

          \node[state, label=above:{\shortstack[l]{$a_1: R = -1$\\$a_2 : R = 1$}},
                right=5cm of S1] (S2) {S$_2$};

          \node[state, label=below:{$a_1, a_2: R = 0$},
                right=3cm of S1, yshift=-2.5cm] (S3) {S$_3$};

          \node[state, label=below:{\shortstack[l]{$a_1: R = -0.1$\\$a_2 : R = 0.5$}},
                below=2cm of S1] (S4) {S$_4$};

          \draw[->] (S1) edge[loop left] node[align=left] {$1/2$} (S1);
          \draw[->] (S2) edge[loop right] node[align=left] {$2/5$} (S2);

          \draw[->, bend left]
            (S1) to node[align=left, yshift=0.5cm] {$1/2$} (S2);

          \draw[->, bend left]
            (S2) to node[align=left, yshift=0.5cm] {$2/5$} (S1);

          \draw[->, bend left]
            (S2) to node[align=left, yshift=-0.5cm] {$1/5$} (S3);

          \draw[->, bend left]
            (S3) to node[align=left, yshift=-0.5cm] {$1/2$} (S1);

          \draw[->, bend left]
            (S3) to node[align=left, yshift=-0.5cm] {$1/2$} (S4);

          \draw[->]
            (S4) to node[left] {$1$} (S1);
        \end{tikzpicture}
    \end{subfigure}

    \caption{Examples illustrating the impact of history-dependent policies.}
    \label{fig:historymatters_combined}
\end{figure}

This setting is already complex enough for it to become difficult to analytically determine the best Markovian policy, not to mention the best history dependent policy. However, we can assume that in a Markovian equilibrium, the adversary would focus its manipulations on $s_1$ and $s_2$, as an incorrect action in $s_3$ or $s_4$ simply has less of a negative effect on the agent. In a history dependent setting, however, the presence of $s_4$ allows the adversary to reduce the certainty the agent had in the previous example of transitioning from $s_3$ to $s_1$. Thus, we would expect the adversary to increase the misreports between $s_1$ and $s_4$ in a history dependent environment. The inherent complexity of even this small MDP and the interplay between agent and attacker policy motivate the following study of an algorithm to find approximate history dependent Nash equilibria. 

Allowing for history dependent strategies vastly increases the strategy space and brings new analytical challenges. As shown in Theorem \ref{thm:nonexist}, history dependent state-robust Nash equilibrium policies (i.e. policies that are optimal for every state) do not always exist.  Therefore, we seek a Nash equilibrium strategy profile for a given initial state distribution. To compute such an equilibrium under history dependence, we modify an algorithm developed by \citet{horak2023solving}. The authors introduce an HSVI-based solver for one-sided partially observable stochastic games, where one player (the attacker) has full state information and the other (the agent) has only partial observability. This algorithm can efficiently approximate equilibrium strategies even in long-horizon problems. Our attacker–agent interaction fits the one-sided paradigm: the attacker knows the true state, while the agent only sees perturbed observations. By leveraging Horak's algorithm, we can obtain a history dependent equilibrium policy pair for our game given the initial state distribution.

\subsection{Modified Zero-Sum OS-POSG to Find Epsilon-Equilibria}\label{sec:modified}

We now state one of our main theorems: that any SA-MDP can be transformed into a strategically equivalent zero-sum one-sided POSG.

\begin{theoremE}[][restate, text link={See \hyperref[proof:prAtEnd\pratendcountercurrent]{remaining details of the proof} on page~\pageref{proof:prAtEnd\pratendcountercurrent}.}]
\label{thm:strategic-equivalence}
Any SA-MDP $G_{\text{seq}}$ can be transformed into a strategically equivalent constrained zero-sum one-sided POSG $\widehat{G}$ (\cref{def:constrained-zoposg}) in which both players act simultaneously at every stage. The solution strategy to $\widehat{G}$ can be mapped back to a unique solution strategy of $G_{\text{seq}}$.
\end{theoremE}
\begin{proof}
We present a high level presentation of the proof, defining its key mathematical objects in the text and its details in the appendix. The underlying idea in the transformation is to encode each sequential SA-MDP as two simultaneous micro-stages while removing all non-decisive degrees of freedom. We introduce alternating dummy actions for both the agent and the adversary. Thus, each player will still have to choose an action at every stage, but will only need to make meaningful decisions every other stage, maintaining the nature of sequential decisions within a simultaneous move game. In order to accommodate the dummy actions, we also introduce dummy states, which act as placeholders while the agent takes the meaningful action that will actually inform the transition probabilities.

\subsubsection*{Defining the Transformation}
We begin with a sequential observation space attack game  $G_\mathrm{seq} = (S, A_1,R,p,\gamma,B)$
and construct a constrained zero-sum OS-POSG 
$\widehat{{G}}=(\widehat{S},\,\widehat{A}_{1},\,\widehat{A}_{2},\,\widehat{O},\,\widehat{T},\,\widehat{R},\,\sqrt{\gamma}, L_1, L_2)$
that is strategically equivalent to $G_{\mathrm{seq}}$ but meets the \emph{simultaneous} action requirements of \citet{horak2023solving}. We define the components of $\widehat{G}$ with the components of $G_{\mathrm{seq}}$ and necessary dummy actions and states. 

\paragraph{State space.} For every $s\in S$ create two states, an \emph{original} state $s$ and a \emph{dummy} state $\Bar{s}$. Hence,
$\widehat{S}=S\cup \Bar{S}, \Bar{S}=\{\Bar{s}:s\in S\}.$

\paragraph{Action sets.}
Introduce distinguished dummy actions $\bot_1 \notin A_1$ for the agent and
$\bot_2 \notin S$ for the adversary, and define
$  \widehat{A}_{1} = A_1\cup\{\bot_1\}, \widehat{A}_{2} = S\cup\{\bot_2\}.$

Agent 1 represents the agent, so its meaningful actions are the original agent actions in $A_1$. Agent 2 represents the attacker, whose meaningful actions encode the observation reported to the agent and therefore range over $S$. The dummy actions $\bot_1,\bot_2$ are used only on non-decisive micro-stages. To eliminate spurious signaling opportunities, we impose state-dependent feasible action sets:
$L_1(s)=\{\bot_1\}, L_2(s)=B(s)$
for each original state $s\in S$, and $ L_1(\bar s)=A_1, L_2(\bar s)=\{\bot_2\} $
for each dummy state $\bar s\in \bar S$.

\paragraph{Stage decomposition.} One stage of ${G}_{\mathrm{seq}}$ is encoded as \emph{two} micro-stages in $\widehat{{G}}$:
\begin{enumerate}
  \item Suppose the system is in an original state $s$. The attacker chooses a feasible action $a_2 \in L_2(s)=B(s)$, which corresponds to the observation it wants the agent to receive, while the agent is forced to play $\bot_1$. The system deterministically transitions to $\bar s$ and emits observation $o=a_2$.

  \item The system is now in the dummy state $\bar s$. The agent chooses $a_1 \in L_1(\bar s)=A_1$ given the received observation $o=a_2$, while the attacker is forced to play $\bot_2$. The system then transitions according to the original probability $p(\cdot\mid s,a_1)$ to the next original state $s'$ and emits a special dummy observation $\bar o$.
\end{enumerate}

\paragraph{Transition Kernel $\widehat{T}$.}
For every original state $s\in S$ and feasible attacker action $a_2\in B(s)$,
$\widehat{T}(a_2,\bar s \mid s,\bot_1,a_2)=1.$ For every dummy state $\bar s\in \bar S$ and agent action $a_1\in A_1$,
$
\widehat{T}(\bar o,s'\mid \bar s,a_1,\bot_2)=p(s'\mid s,a_1).
$
All other transitions have probability zero (equivalently, they correspond to infeasible action tuples under the state-dependent constraints above).

Note that the first micro-stage depends only on the attacker’s meaningful action, while the second micro-stage depends only on the agent's meaningful action. The dummy actions never affect transitions, rewards, or observations.

\paragraph{Reward.} Define 
  $\widehat{R}(s,{a}_{1},a_{2}) = 0, \widehat{R}(\bar{s},a_{1},{a}_{2}) = \frac{1}{\sqrt{\gamma}}R(s,a_{1})$
where the $\frac{1}{\sqrt{\gamma}}$ accounts for the rewards being applied at odd indices.

\paragraph{Observations.} The agent's observation alphabet is $\widehat{O}=S\cup\{\bar{o}\}$. Player~1 receives the attacker-selected observation $a_2$ at the first micro-stage and the dummy observation $\bar{o}$ at the second micro-stage. The observation $\bar{o}$ signals that the next micro-stage is attacker-decisive and the agent is constrained to play $\bot_1$. Player~2 always observes the full state.

\paragraph{Discount Factor.}
The transformed game has discount factor $\sqrt{\gamma}$. Since the only non-zero rewards occur at the second micro-stage of each encoded step, we scale rewards on dummy states by $1/\sqrt{\gamma}$. Hence
    $(\sqrt{\gamma})^{2t+1}\widehat R(\bar s,a_1,\bot_2) 
    = (\sqrt{\gamma})^{2t+1}\frac{1}{\sqrt{\gamma}}R(s,a_1) = \gamma^t R(s,a_1)$,   
so the transformed discounted return matches the original discounted return exactly.

Consider the following visualization (\autoref{fig:classicformulation}) of how one stage in the  sequential observation space attack game can be simulated by a simultaneous move game. We will consider the first stage of the game, transitioning from $s^{0}$ to either $s_1$ or $s_2$. 

\begin{figure}
    \centering

    \begin{subfigure}[t]{0.49\linewidth}
        \centering
        \caption{Sequential Move Formulation}
        \label{fig:classicformulation}

        \begin{tikzpicture}[
            scale=0.60, transform shape,
            ->, >=Stealth,
            node distance=2.5cm,
            every state/.style={draw, circle, minimum size=6mm, font=\small}
          ]
          \node[state] (s0) {$\mathbf{s^{0}}$};
          \node[state, right=5cm of s0, yshift=2cm] (s1) {$\mathbf{s_1}$};
          \node[state, right=5cm of s0, yshift=-2cm] (s2) {$\mathbf{s_2}$};

          \node[state, right=1cm of s0, draw=black, fill=white] (o0) {$o^t$};

          \node[state, above=1cm of o0, draw=black, fill=white] (aa0) {$a^t_2$};

          \coordinate (o1) at ($(s1)+(1cm,0)$);
          \coordinate (o2) at ($(s2)+(1cm,0)$);

          \coordinate (s3) at ($(s1)+(5cm, 1.5cm)$);
          \coordinate (s5a) at ($(s1)+(5cm,-2cm)$);
          \coordinate (s4) at ($(s2)+(5cm,-1.5cm)$);
          \coordinate (s5b) at ($(s2)+(5cm, 0cm)$);

          \path
            (s0) edge (o0)
            (s1) edge (o1)
            (s2) edge (o2);

          \path (aa0) edge (o0);

          \path
            (o0) edge node[above, font=\footnotesize, align=center] {
              $a_1=0.5$\\
              $a_2=0.7$
            } (s1)
                 edge node[below, font=\footnotesize, align=center] {
              $\quad$\\
              $a_1=0.5$\\
              $a_2=0.3$
            } (s2);
        \end{tikzpicture}
    \end{subfigure}
    \hfill
    \begin{subfigure}[t]{0.49\linewidth}
        \centering
        \caption{Simultaneous Move Formulation}
        \label{fig:zoposgformulation}

        \begin{tikzpicture}[
            scale = 0.60, transform shape,
            ->, >=Stealth,
            node distance=2.5cm,
            every state/.style={draw, circle, minimum size=6mm, font=\small}
          ]

          \node[state]                   (s0) {$\mathbf{s^{0}}$};
          \node[state, right=9cm of s0, yshift=1.5cm] (s1) {$\mathbf{s_1}$};
          \node[state, below=of s1]     (s2) {$\mathbf{s_2}$};

          \node[state, right=1cm of s0, draw=black, fill=white] (o0) {$o^{t}$};
          \node[state, left=1cm of s1, draw=black, fill=white] (o1) {$\Bar{o}$};
          \node[state, left=1cm of s2, draw=black, fill=white] (o2) {$\Bar{o}$};

          \node[state, right=1cm of o0, draw=black, fill=white] (s0') {$\mathbf{s^{0,d}}$};

          \path
            (s0) edge node[above, font=\footnotesize, align=center] {
              $(*, a_a)$
            } (o0);

          \path
            (o0) edge (s0');

          \path
            (o1) edge (s1)
            (o2) edge (s2);

          \path
            (s0') edge node[above, font=\footnotesize, align=center] {
              $(a_1, *)=0.5$\\
              $(a_2, *)=0.7$\\
            } (o1)
                 edge node[below, font=\footnotesize, align=center] {
              $\quad$\\
              $(a_1,*)=0.5$\\
              $(a_2,*)=0.3$
            } (o2);

        \end{tikzpicture}
    \end{subfigure}

    \caption{Strategically equivalent forms of an OS-POSG.}
\end{figure}

Note that in the sequential setting, the observation emitted by the system is modified according to the attacker's action. This observation informs the agent's action, while the true state dictates the actual transition probabilities.

\subsubsection*{Strategic Equivalence} We now show that the transformed simultaneous zero-sum OS-POSG $\widehat G$ is strategically equivalent to the original sequential observation-space attack game $G_{\mathrm{seq}}$. We establish the following three properties: \begin{enumerate} 
\item \emph{Strategy correspondence:} after restricting each non-decisive micro-stage to its singleton dummy action, there is a bijection \[ \Phi_i:\Pi_i \longrightarrow \Sigma_i^{\mathrm{feas}}, \qquad i\in\{1,2\}, \] between each player's behavioral strategies in $G_{\mathrm{seq}}$ and feasible behavioral strategies in $\widehat G$.
\item \emph{Pairwise payoff preservation:} for every initial belief $b$, its canonical embedding $\widehat b$ ($\widehat b(s)=b(s)$ for $s\in S$, $\widehat b(\bar s)=0$) in the transformed state space, and every strategy pair $(\pi_1,\pi_2)$, \[ V(b;\pi_1,\pi_2) = \widehat V \bigl( \widehat b; \Phi_1(\pi_1),\Phi_2(\pi_2) \bigr). \]
\item \emph{$\epsilon$-equilibrium correspondence:} a strategy pair is an $\epsilon$-Nash equilibrium, equivalently an $\epsilon$-saddle point, of $G_{\mathrm{seq}}$ if and only if its image under $(\Phi_1,\Phi_2)$ is an $\epsilon$-Nash equilibrium of $\widehat G$ over feasible strategies. \end{enumerate}
These three results prove \cref{thm:strategic-equivalence}. Result 1 guarantees that a strategy found in $\widehat{G}$ will have a corresponding unique strategy in $G_{\text{seq}}$. Result 2 guarantees that the values of the Nash equilibria in both games are equal. Finally, Result 3 guarantees that the strategy in $G_{\text{seq}}$ still achieves an $\epsilon$-approximate Nash equilibrium.

\begin{proofE}
After establishing the key mathematical objects and statements to prove in the main text, we now proceed with the actual proof.
\subsubsection*{Preliminaries and Definitions}
Let $\pi_1$ (resp. $\pi_2$) denote the agent (resp. adversary) strategy in the sequential game $G_{\mathrm{seq}}$, with value
$$
  V(b;\pi_1,\pi_2)
  =\mathbb{E}_{s^{0}\sim b}\Bigl[ \sum_{t=0}^\infty \gamma^t R\bigl(s^t,a_1^t\bigr) \Bigr],
$$
and let $\sigma_1,\sigma_2$ be analogous strategies in the simultaneous game $\widehat G$ with value
$$
  \widehat V(\widehat b;\sigma_1,\sigma_2)
  =\mathbb{E}_{x^{0}\sim\widehat b}\Bigl[ \sum_{n=0}^\infty (\sqrt\gamma)^{n} \widehat R\bigl(x^n,a_1^n,a_2^n\bigr)\Bigr],
$$
where $x^n,a^n_i$  denote states/actions at micro--step.

Define a bijection $\iota:(a_2^0,o^0,a_1^0,a_2^1,o^1,a_1^1,\dots)
\rightarrow
(\bot_1,a_2^0,o^0,a_1^0,\bot_2,\bar o,\bot_1,a_2^1,o^1,a_1^1,\bot_2,\bar o,...)$
that interleaves each stage of $G_{\mathrm{seq}}$ into two constrained simultaneous micro-stages of $\widehat G$: at step $2t$ the attacker chooses $a_2^t$ while the agent is forced to play $\bot_1$, and at step $2t+1$ the agent chooses $a_1^t$ while the adversary is forced to play $\bot_2$. Because the non-decisive actions are fixed and observations are emitted in the same order as in the sequential game, the mapping $\iota$ is bijective and preserves each player's information set.

Given a sequential strategy pair $(\pi_1,\pi_2)$, define the feasible simultaneous strategy pair $(\sigma_1,\sigma_2)$ by
\begin{align*}
\sigma_1(\cdot\mid\iota(h^{\mathrm{seq}}),o^{2t})=\pi_1(\cdot\mid h^{\mathrm{seq}},o^t),\quad\sigma_2(\cdot\mid\iota(h^{\mathrm{seq}}),x^{2t})=\pi_2(\cdot\mid h^{\mathrm{seq}}),
\end{align*}
with
\begin{align*}
\sigma_1(\cdot\mid \iota(h^{\mathrm{seq}})) =\delta_{\bot_1}, \qquad \sigma_2(\cdot\mid\iota(h^{\mathrm{seq}}),x^{2t+1})=\delta_{\bot_2}.
\end{align*}

Here $h^{\mathrm{seq}}$ denotes the relevant player’s private history. 
For player 1, $\iota$ inserts the forced dummy components into its action-observation history, while for player 2, it inserts the dummy states and actions into its full state-action history.

Conversely, every feasible simultaneous strategy pair induces a unique sequential strategy pair by deleting the forced dummy moves.

\begin{definition}[$\epsilon$--equilibrium concepts]
\label{def:epsilon-equilibria}
Recall that player $1$ is the maximizing player and player $2$ is the
minimizing player.

A feasible strategy pair $(\sigma_1^*,\sigma_2^*)\in
\Sigma_1^{\mathrm{feas}}
\times
\Sigma_2^{\mathrm{feas}}$
is an $\epsilon$--Nash equilibrium, equivalently an
$\epsilon$--saddle point, of $\widehat G$ if
\begin{align}
\widehat V
(\widehat b;\sigma_1,\sigma_2^*)
&\le
\widehat V
(\widehat b;\sigma_1^*,\sigma_2^*)
+\epsilon
&&
\forall\,
\sigma_1\in\Sigma_1^{\mathrm{feas}},
\label{eq:sim-eps-player1}
\\
\widehat V
(\widehat b;\sigma_1^*,\sigma_2)
&\ge
\widehat V
(\widehat b;\sigma_1^*,\sigma_2^*)
-\epsilon
&&
\forall\,
\sigma_2\in\Sigma_2^{\mathrm{feas}}.
\label{eq:sim-eps-player2}
\end{align}

A strategy pair
$(\pi_1^*,\pi_2^*)\in
\Pi_1\times\Pi_2$
is an $\epsilon$--Nash equilibrium, equivalently an
$\epsilon$--saddle point, of $G_{\mathrm{seq}}$ if
\begin{align}
V(b;\pi_1,\pi_2^*)
&\le
V(b;\pi_1^*,\pi_2^*)
+\epsilon
&&
\forall\,\pi_1\in\Pi_1,
\label{eq:seq-eps-player1}
\\
V(b;\pi_1^*,\pi_2)
&\ge
V(b;\pi_1^*,\pi_2^*)
-\epsilon
&&
\forall\,\pi_2\in\Pi_2.
\label{eq:seq-eps-player2}
\end{align}
\end{definition}

\subsubsection*{Theorem Statements and Proofs}

\begin{lemma}[Exact value preservation]
\label{thm:exactvaluepreservation}
Under optimal play,
\begin{align*}
  \max_{\pi_1}\min_{\pi_2} V(b;\pi_1,\pi_2) = \max_{\sigma_1}\min_{\sigma_2} \widehat V (\widehat b;\sigma_1,\sigma_2).
\end{align*}
\end{lemma}
\begin{proof}
We use the bijection $\iota$ to translate any sequential strategy pair $(\pi_1,\pi_2)$ into a feasible simultaneous pair $(\sigma_1,\sigma_2)$ and vice-versa, without altering the induced probability measure over outcome streams. Observe:
\begin{itemize}
\item[1.] \emph{Reward alignment:} by construction, the only non-zero reward in the transformed game occurs at the second micro-stage $x^{2t+1}$, and
\[
\widehat R(x^{2t+1},a_1^{2t+1},a_2^{2t+1})
=
\frac{1}{\sqrt{\gamma}}R(s^t,a_1^t).
\]

\item[2.] \emph{Discount consistency:} therefore
\begin{align*}
(\sqrt\gamma)^{2t+1}\widehat R(x^{2t+1},a_1^{2t+1},a_2^{2t+1})
    = 
(\sqrt\gamma)^{2t+1}\frac{1}{\sqrt{\gamma}}R(s^t,a_1^t)
    = 
\gamma^t R(s^t,a_1^t).
\end{align*}

\item[3.] \emph{Measure preservation:} at each original stage, corresponding strategies assign the same probabilities to the attacker’s report and the agent’s action. The inserted moves are deterministic, and the second micro-stage uses the same transition probability \(p(s^{t+1}\mid s^t,a_1^t)\). Hence, after deleting the dummy components, corresponding trajectories have the same probability.

\item[4.] \emph{Feasibility preservation:} the singleton constraints ensure that the transformed game introduces no additional strategic choices beyond those already present in $G_{\mathrm{seq}}$.
\end{itemize}
Therefore, for any corresponding strategy pair,
\[
V(b;\pi_1,\pi_2)=\widehat V(\widehat b;\sigma_1,\sigma_2).
\]
Taking max--min over admissible strategies yields the desired equality.
\end{proof}

\begin{lemma}[$\epsilon$-equilibrium correspondence] \label{lem:epsilon-equilibrium-correspondence} Fix an initial belief $b$. For each player $i\in\{1,2\}$, let \[ \Phi_i:\Pi_i \longrightarrow \Sigma_i^{\mathrm{feas}} \] be a bijection between the behavioral strategies in $G_{\mathrm{seq}}$ and the feasible behavioral strategies in $\widehat G$. Suppose that, for every $(\pi_1,\pi_2)$, \[ V(b;\pi_1,\pi_2) = \widehat V \bigl(\widehat b;\Phi_1(\pi_1),\Phi_2(\pi_2)\bigr). \] Then $(\pi_1^*,\pi_2^*)$ is an $\epsilon$-Nash equilibrium of $G_{\mathrm{seq}}$ if and only if $\bigl(\Phi_1(\pi_1^*),\Phi_2(\pi_2^*)\bigr) $ is an $\epsilon$-Nash equilibrium of $\widehat G$. \end{lemma}

\begin{proof} Write \[\sigma_i^*=\Phi_i(\pi_i^*),\qquad i\in\{1,2\}. \]
Assume first that $(\pi_1^*,\pi_2^*)$ is an $\epsilon$-Nash equilibrium of $G_{\mathrm{seq}}$.
For any feasible deviation $\sigma_1\in\Sigma_1^{\mathrm{feas}}$, bijectivity gives $\pi_1=\Phi_1^{-1}(\sigma_1)$.
Hence,
\begin{align*} 
\widehat V(\widehat b;\sigma_1,\sigma_2^*) &= V(b;\pi_1,\pi_2^*)\\
&\le V(b;\pi_1^*,\pi_2^*)+\epsilon\\
&= \widehat V(\widehat b;\sigma_1^*,\sigma_2^*)+\epsilon.
\end{align*} 
Similarly, for any feasible deviation $\sigma_2\in\Sigma_2^{\mathrm{feas}}$, setting $\pi_2=\Phi_2^{-1}(\sigma_2)$ yields
\begin{align*} 
\widehat V(\widehat b;\sigma_1^*,\sigma_2) &= V(b;\pi_1^*,\pi_2)\\
&\ge V(b;\pi_1^*,\pi_2^*)-\epsilon\\
&= \widehat V(\widehat b;\sigma_1^*,\sigma_2^*)-\epsilon.
\end{align*} Thus, $(\sigma_1^*,\sigma_2^*)$ is an $\epsilon$-Nash equilibrium of $\widehat G$. The converse follows by the same argument, applying the inverse bijections $\Phi_i^{-1}$ to arbitrary deviations in $G_{\mathrm{seq}}$. Finally, the feasible strategy spaces of $\widehat G$ admit no additional deviations at dummy micro-stages, since the acting player has a singleton feasible action set there. Therefore, every feasible deviation in $\widehat G$ corresponds uniquely to a deviation in $G_{\mathrm{seq}}$, as required.
\end{proof}

\begin{lemma}
    [Well-defined strategy correspondence] The bijection $\iota$ ensures a one-to-one correspondence between strategies in $\widehat G$ and strategies in $G_{\text{seq}}$.
\end{lemma}

\begin{proof}
    By construction of $\iota$.
\end{proof}
This concludes the proof of \cref{thm:strategic-equivalence}.
\end{proofE} 
\end{proof}

\subsection{Constrained Zero-sum OS-POSG}

\begin{definition}[Information-compatible constrained zero-sum OS-POSG]
\label{def:constrained-zoposg}
An information-compatible constrained zero-sum one-sided partially
observable stochastic game is a tuple
\[
G=(S,A_1,A_2,O,T,R,\gamma,L_1,L_2),
\]
where $(S,A_1,A_2,O,T,R,\gamma)$ is a finite discounted zero-sum
one-sided POSG, player $1$ is the uninformed maximizing player, and
player $2$ is the fully informed minimizing player.

For player $2$,
\[
L_2:S\to 2^{A_2}\setminus\{\emptyset\}
\]
may depend on the true state. For player $1$, the feasible action set is
required to be measurable with respect to player $1$'s information:
whenever two states $s,s'$ may occur at the same player-$1$ information
history,
\[
L_1(s)=L_1(s').
\]
Equivalently, at every reachable player-$1$ information history $h_1$,
there is a nonempty set $L_1(h_1)\subseteq A_1$ that is common to all
states compatible with $h_1$.

A behavioral strategy is feasible if its support is contained in the
corresponding feasible action set at every information history.
\end{definition}

The game is played as before, with the additional requirement that, at player 1 history $h_1$, player 1 chooses $a_1\in L_1(h_1)$, and, at state $s$, player 2 chooses $a_2\in L_2(s)$.
\begin{theoremE}[][restate,text link section]
\label{thm:constrainedZOPOSG}
The HSVI algorithm of \citet{horak2023solving} extends to the constrained zero-sum one-sided POSG in \cref{def:constrained-zoposg} by restricting each local stage game to the players' feasible strategy sets. The algorithm maintains a lower-bound value function $V_{LB}$ and an upper-bound value function $V_{UB}$, written as $V^{\Gamma}_{LB}$ and $V^{\Upsilon}_{UB}$, respectively, in \cref{alg:constrained_hsvi}, satisfying
\[
V_{LB}(b)\leq V^*(b)\leq V_{UB}(b)
\]
for every reachable belief $b$, where $V^*$ is the optimal value function of the constrained game.

Let
\[
\delta=\frac{U_0-L_0}{2},
\]
where $L_0$ and $U_0$ are the discounted minimum and maximum feasible reward
bounds. If $\delta=0$, the initial bounds coincide. If $\delta>0$ and
\[
0<D<\frac{(1-\gamma)\epsilon}{2\delta},
\]
then the algorithm terminates at the initial belief $b^0$ with
\[
V_{UB}(b^0)-V_{LB}(b^0)\leq\epsilon.
\]
The strategy obtained from the lower-bound value function for player~$1$ and the strategy
obtained from the upper-bound value function for player~$2$ are feasible and form an
$\epsilon$-Nash equilibrium. In particular, the result applies to the
constrained OS-POSG $\widehat G$ constructed in
\cref{thm:strategic-equivalence}.
\end{theoremE}
\begin{proofE}
We adapt the proof of \citet{horak2023solving} by restricting the action
sets in each local stage game. We state the parts of their argument that are
unchanged and make explicit where the feasibility restrictions enter.

\paragraph{Model and feasible stage strategies.}
At a player-$1$ information history $h_1$, let $S(h_1)$ be the states that
may be current at that history. By the condition in
\cref{def:constrained-zoposg}, all states in $S(h_1)$ have the common
feasible action set $L_1(h_1)$. Thus, at a belief
$b\in\Delta(S(h_1))$, the stage strategies are
\[
\pi_1\in\Pi_1(h_1)\coloneqq\Delta(L_1(h_1)),
\qquad
\pi_2\in\Pi_2(h_1)\coloneqq
\prod_{s\in S(h_1)}\Delta(L_2(s)).
\]
Player~$1$ therefore uses one distribution over $L_1(h_1)$; only
player~$2$ may condition its action on the current state. Both strategy sets
are nonempty, compact, and convex.

\paragraph{Belief update.}
Fix $b\in\Delta(S(h_1))$, $a_1\in L_1(h_1)$, and
$\pi_2\in\Pi_2(h_1)$. Define
\begin{align*}
p(o,s'\mid b,a_1,\pi_2)
&=
\sum_{s\in S(h_1)}\sum_{a_2\in L_2(s)}
b(s)\pi_2(a_2\mid s)T(o,s'\mid s,a_1,a_2),\\
P(o\mid b,a_1,\pi_2)
&=\sum_{s'\in S}p(o,s'\mid b,a_1,\pi_2).
\end{align*}
For the successor history $h_1'=(h_1,a_1,o)$, the states that may be current
are
\[
S(h_1')=
\left\{
s'\in S:
T(o,s'\mid s,a_1,a_2)>0
\text{ for some }s\in S(h_1),\ a_2\in L_2(s)
\right\}.
\]
Action-observation pairs for which $S(h_1')=\emptyset$ are omitted from all
continuation sums and lower-bound constraints, since their probability is
zero under every feasible stage strategy.
Thus $S(h_1')$ depends only on $S(h_1)$, $a_1$, and $o$, and every posterior
generated by a feasible $\pi_2$ on a remaining branch is supported on
$S(h_1')$.
When $P(o\mid b,a_1,\pi_2)>0$, player~$1$'s posterior after observing its
action $a_1$ and observation $o$ is
\begin{equation}
\label{eq:constrained-posterior}
\tau(b,a_1,\pi_2,o)(s')
=
\frac{p(o,s'\mid b,a_1,\pi_2)}
{P(o\mid b,a_1,\pi_2)}.
\end{equation}
On a zero-probability branch, the posterior may be chosen arbitrarily in
$\Delta(S(h_1'))$. If
player~$1$ uses $\pi_1$, the probability of the branch $(a_1,o)$ is
\[
\Pr_{b,\pi_1,\pi_2}[a_1,o]
=\pi_1(a_1)P(o\mid b,a_1,\pi_2).
\]
The factor $\pi_1(a_1)$ does not appear in
\eqref{eq:constrained-posterior} because player~$1$ conditions on the action
it took. The condition in \cref{def:constrained-zoposg} also ensures that
every positive-probability successor belief is supported on a set of states
with one common feasible action set for player~$1$.

\paragraph{Values and structure.}
For a feasible strategy $\sigma_1$ of player~$1$, define
\[
\underline V(b;\sigma_1)
=\inf_{\sigma_2}
\mathbb E_{b,\sigma_1,\sigma_2}
\left[\sum_{t\geq0}\gamma^{t}R(s^t,a^t_1,a^t_2)\right],
\qquad
V^*(b)=\sup_{\sigma_1}\underline V(b;\sigma_1),
\]
where both optimizations are over feasible behavioral strategies. Since
player~$2$ observes the state, a best response to a fixed $\sigma_1$ can be
chosen separately for each initial state. Hence $\underline V(\cdot;\sigma_1)$ is affine in
$b$ on each $\Delta(S(h_1))$, and $V^*$ is convex there.

Let
\begin{align*}
R_{\min}
&=\min_{s\in S,\,a_1\in L_1(s),\,a_2\in L_2(s)}R(s,a_1,a_2),\\
R_{\max}
&=\max_{s\in S,\,a_1\in L_1(s),\,a_2\in L_2(s)}R(s,a_1,a_2),\\
L_0&=\frac{R_{\min}}{1-\gamma},
\qquad
U_0=\frac{R_{\max}}{1-\gamma},
\qquad
\delta=\frac{U_0-L_0}{2}.
\end{align*}
Then $L_0\leq \underline V(b;\sigma_1)\leq U_0$. As in
\citet{horak2023solving}, every $\underline V(\cdot;\sigma_1)$, and therefore $V^*$, is
$\delta$-Lipschitz in $\lVert\cdot\rVert_1$ on each
$\Delta(S(h_1))$.

\paragraph{Bellman operator.}
Let $V$ be a bounded convex continuous continuation value on the reachable
beliefs. At $b\in\Delta(S(h_1))$, define
\begin{align*}
u_{V,b}(\pi_1,\pi_2)
&=\mathbb E_{b,\pi_1,\pi_2}[R(s,a_1,a_2)]\\
&\quad+\gamma\sum_{a_1\in L_1(h_1)}\sum_{o\in O}
\Pr_{b,\pi_1,\pi_2}[a_1,o]
V\bigl(\tau(b,a_1,\pi_2,o)\bigr),\\
[HV](b)
&=\max_{\pi_1\in\Pi_1(h_1)}
\min_{\pi_2\in\Pi_2(h_1)}u_{V,b}(\pi_1,\pi_2),
\end{align*}
where zero-probability branches contribute zero. For fixed $\pi_2$, this
payoff is affine in $\pi_1$. For fixed $\pi_1$, each continuation term is
the perspective of the convex function $V$ and is therefore continuous and
convex in $\pi_2$. Sion's minimax theorem gives
\[
\max_{\pi_1\in\Pi_1(h_1)}\min_{\pi_2\in\Pi_2(h_1)}u_{V,b}(\pi_1,\pi_2)
=
\min_{\pi_2\in\Pi_2(h_1)}\max_{\pi_1\in\Pi_1(h_1)}u_{V,b}(\pi_1,\pi_2).
\]
Thus the local stage game has a value and feasible optimal strategies.
The value-composition argument of \citet{horak2023solving} applies on each
$\Delta(S(h_1))$. It follows that $H$ maps convex $\delta$-Lipschitz
continuation values to convex $\delta$-Lipschitz values on the current
simplex.

For any bounded $V$ and $W$,
\[
\lVert HV-HW\rVert_\infty
\leq\gamma\lVert V-W\rVert_\infty,
\]
where the supremum is over all reachable beliefs. The proof is unchanged:
for fixed stage strategies, replacing $V$ by $W$ changes the expected
continuation value by at most $\gamma\lVert V-W\rVert_\infty$, and taking
the maximum and minimum preserves the inequality. Moreover, every feasible
strategy of player~$1$ after $h_1$ consists of a feasible distribution in
$\Delta(L_1(h_1))$ and feasible continuation strategies after each
$(a_1,o)$, and every such choice defines a feasible behavioral strategy.
The strategy-decomposition argument of \citet{horak2023solving} therefore
gives $HV^*=V^*$. Hence $H$ has a unique fixed point, which is $V^*$.

\paragraph{Lower-bound stage program.}
The lower bound is represented by sets of $\alpha$-vectors, one for each
distinct set $S(h_1)$. Consider a point update at
$b\in\Delta(S(h_1))$. For each $(a_1,o)$, write
$h_1'=(h_1,a_1,o)$ for the successor information history, and let
$\Gamma_{a_1,o}=\{\alpha_1^{a_1,o},\ldots,
\alpha_{k_{a_1,o}}^{a_1,o}\}$ be the current lower-bound vectors on that
successor set $S(h_1')$. The lower-bound program of
\citet{horak2023solving} becomes
\begin{align*}
\max_{\pi_1,\widehat\lambda,\widehat\alpha,V_s}\quad
&\sum_{s\in S(h_1)}b(s)V_s\\
\text{s.t.}\quad
V_s\leq{}&
\sum_{a_1\in L_1(h_1)}\pi_1(a_1)R(s,a_1,a_2)\\
&+\gamma\sum_{a_1\in L_1(h_1)}\sum_{o\in O}
\sum_{s'\in S(h_1')}
T(o,s'\mid s,a_1,a_2)\widehat\alpha_{a_1,o}(s')
&&\substack{\forall s\in S(h_1),\\
a_2\in L_2(s),}\\
\widehat\alpha_{a_1,o}(s')={}&
\sum_{i=1}^{k_{a_1,o}}
\widehat\lambda_i^{a_1,o}\alpha_i^{a_1,o}(s')
&&\substack{\forall a_1\in L_1(h_1),\ o\in O,\\
s'\in S(h_1'),}\\
\sum_{i=1}^{k_{a_1,o}}\widehat\lambda_i^{a_1,o}={}&
\pi_1(a_1)
&&\substack{\forall a_1\in L_1(h_1),\\o\in O,}\\
\sum_{a_1\in L_1(h_1)}\pi_1(a_1)={}&1,
\qquad \pi_1\geq0,
\qquad\widehat\lambda\geq0.
\end{align*}
For every $a_1$ with $\pi_1(a_1)>0$, set
\[
\alpha_{a_1,o}
=\frac{\widehat\alpha_{a_1,o}}{\pi_1(a_1)}
\in\operatorname{conv}(\Gamma_{a_1,o}).
\]
When $\pi_1(a_1)=0$, choose $\alpha_{a_1,o}$ arbitrarily from
$\operatorname{conv}(\Gamma_{a_1,o})$. The vector added to the lower bound
is the one-step vector obtained from $\pi_1$ and these continuation vectors.
Equivalently, it is defined on every $s\in S(h_1)$ by the minimum of the
right-hand side of the first constraint over $a_2\in L_2(s)$. This fixes
every coordinate of the new vector, including states to which $b$ assigns
zero probability, and its value at $b$ equals the optimal objective of the
program.

The only changes are the restrictions $a_1\in L_1(h_1)$,
$a_2\in L_2(s)$, and the use of the lower-bound vectors for the successor
belief. The vector
$\widehat\alpha_{a_1,o}$ is already multiplied by the probability
$\pi_1(a_1)$, since its mixture weights sum to $\pi_1(a_1)$; it is
therefore not multiplied by $\pi_1(a_1)$ again. The dual is the dual of
\citet{horak2023solving} with one realization weight
$b(s)\pi_2(a_2\mid s)$ for every $a_2\in L_2(s)$. Both programs are finite,
feasible, and bounded, so strong duality holds.

\paragraph{Bounds and point updates.}
For each distinct $S(h_1)$, initialize the lower bound with the constant
vector $L_0$ and the upper bound with the points $(e_s,U_0)$,
$s\in S(h_1)$. The first is a valid lower bound because no feasible strategy
can receive less than $L_0$. The lower $\delta$-Lipschitz envelope of the
upper points is the constant $U_0$ and is a valid upper bound.

Every vector produced by the lower-bound program combines a feasible local
strategy of player~$1$ with continuation vectors that already lower-bound
the values of feasible continuation strategies. Following those strategies
after the corresponding $(a_1,o)$ therefore gives a feasible behavioral
strategy whose value is at least the new vector. Convex combinations of
stored vectors remain valid because player~$1$ can randomize privately among
the corresponding strategies. It follows by induction that the lower bound
remains valid after every point update.
Moreover, every coordinate of every lower-bound vector remains in
$[L_0,U_0]$. This is true initially, and the one-step construction preserves
these bounds because $R_{\min}+\gamma L_0=L_0$ and
$R_{\max}+\gamma U_0=U_0$. Each vector therefore defines a
$\delta$-Lipschitz linear function, and their pointwise maximum is also
$\delta$-Lipschitz.

For the upper bound, add the point $(b,[HV_{UB}](b))$ and take the same lower
$\delta$-Lipschitz envelope as \citet{horak2023solving}, using only points
on the current $\Delta(S(h_1))$. If the stored points upper-bound $V^*$,
then monotonicity of $H$ gives
\[
[HV_{UB}](b)\geq[HV^*](b)=V^*(b),
\]
so the new point is also valid. Convexity and the $\delta$-Lipschitz
property then show that the updated envelope upper-bounds $V^*$ throughout
that simplex. No interpolation is performed between sets with different
feasible actions. Thus, throughout the algorithm,
\[
V_{LB}(b)\leq V^*(b)\leq V_{UB}(b).
\]

The same updates also preserve the two inequalities used by the strategy
construction of \citet{horak2023solving}. For the initial lower vector, any
feasible $\pi_1$ together with the same constant continuation gives a
one-step value of at least $L_0$. Every subsequently added vector is the
one-step vector obtained from the feasible $\pi_1$ returned by the
lower-bound program and continuation vectors in the corresponding sets
$\operatorname{conv}(\Gamma_{a_1,o})$. The same property holds for convex
combinations. Indeed, if $\alpha=\sum_j\lambda_j\alpha^j$, use
$\pi_1=\sum_j\lambda_j\pi_1^j$ and, after an action $a_1$ with
$\pi_1(a_1)>0$, mix the corresponding continuation vectors with weights
$\lambda_j\pi_1^j(a_1)/\pi_1(a_1)$. This gives a feasible one-step vector
that dominates $\alpha$ coordinatewise. For the upper bound,
\[
[HV_{UB}](b)\leq V_{UB}(b)
\]
holds initially because
$HU_0\leq R_{\max}+\gamma U_0=U_0$, and it continues to hold after every
update. To see the latter, let $V'_{UB}$ be
the envelope after adding a point. Since $V'_{UB}\leq V_{UB}$, every old
point $(b_i,y_i)$ satisfies
\[
y_i\geq V_{UB}(b_i)
\geq[HV_{UB}](b_i)
\geq[HV'_{UB}](b_i),
\]
and the new point satisfies
\[
[HV_{UB}](b)\geq[HV'_{UB}](b).
\]
The convexity and $\delta$-Lipschitz continuity of $HV'_{UB}$ then imply
that its lower
$\delta$-Lipschitz envelope, $V'_{UB}$, also lies above $HV'_{UB}$. This is
the same argument as in \citet{horak2023solving}, applied separately on
each $\Delta(S(h_1))$.

\paragraph{Termination.}
If $\delta=0$, the initial bounds coincide. Suppose $\delta>0$ and define
\[
\rho(0)=\epsilon,
\qquad
\rho(t+1)=\frac{\rho(t)-2\delta D}{\gamma},
\qquad
\operatorname{excess}(b,t)
=V_{UB}(b)-V_{LB}(b)-\rho(t).
\]
Let $\pi_1^{UB}$ be optimal in the upper-bound stage game and
$\pi_2^{LB}$ be optimal in the lower-bound stage game. The saddle-point
inequalities give
\begin{align*}
[HV_{UB}](b)-[HV_{LB}](b)
\leq{}&\gamma\sum_{a_1,o}
\Pr_{b,\pi_1^{UB},\pi_2^{LB}}[a_1,o]\\
&\quad\times
\bigl(V_{UB}-V_{LB}\bigr)
\bigl(\tau(b,a_1,\pi_2^{LB},o)\bigr).
\end{align*}
Consequently, the forward rule of \citet{horak2023solving} is unchanged: it
selects the branch maximizing its probability multiplied by the successor
excess, using the action-conditioned posterior in
\eqref{eq:constrained-posterior}.

If the selected quantity is nonpositive, every positive-probability
successor has gap at most $\rho(t+1)$. After the point update, the preceding
inequality gives
\[
V_{UB}(b)-V_{LB}(b)
\leq\gamma\rho(t+1)
=\rho(t)-2\delta D.
\]
Both bounds are $\delta$-Lipschitz on the current simplex, so their
difference is $2\delta$-Lipschitz. Hence every belief in that simplex within
distance $D$ of $b$ has nonpositive excess at depth $t$. If
\[
0<D<\frac{(1-\gamma)\epsilon}{2\delta},
\]
then $\rho(t)$ eventually exceeds the largest possible bound gap, so the
recursion depth is finite. Fix a depth $t$ and a structural state set
$S(h_1)$, and consider the last belief reached by each trial that terminates
there. After the point update at such a belief $b$, the excess is
nonpositive throughout the $D$-neighborhood of $b$ in
$\Delta(S(h_1))$. Later updates can only reduce the gap, so a later trial
can terminate at that depth and state set only at a belief more than $D$
from $b$. These terminal beliefs are therefore $D$-separated. Compactness
of each $\Delta(S(h_1))$, together with the finite recursion depth and the
finitely many distinct subsets $S(h_1)\subseteq S$, implies that only
finitely many trials occur. Hence the algorithm terminates with
\[
V_{UB}(b^0)-V_{LB}(b^0)\leq\epsilon.
\]

\paragraph{Strategy extraction.}
The strategy construction of \citet{horak2023solving} uses exactly the two
inequalities established above. For player~$1$, begin with a lower-bound
vector $\alpha$ attaining $V_{LB}(b^0)$. At each information history, play
the feasible distribution associated with $\alpha$ and, after observing
$(a_1,o)$, replace $\alpha$ by the associated continuation vector in
$\operatorname{conv}(\Gamma_{a_1,o})$. This construction depends only on
player~$1$'s action-observation history and the retained continuation vector;
it does not require player~$1$ to know the current belief. If it is followed
for $K$ stages and an arbitrary feasible strategy is used thereafter, the
finite-horizon induction of \citet{horak2023solving} gives a value of at
least
\[
V_{LB}(b^0)-\gamma^K(U_0-L_0).
\]
Letting $K\to\infty$ gives a feasible strategy $\sigma_1^{LB}$ such that
\[
\inf_{\sigma_2}V(b^0;\sigma_1^{LB},\sigma_2)
\geq V_{LB}(b^0).
\]
For player~$2$, choose at each stage the minimizing strategy in
$[HV_{UB}](b)$. Player~$2$ observes the current state and knows the stage
strategies it has used, so it can update $b$ by
\eqref{eq:constrained-posterior}. The inequality
$HV_{UB}\leq V_{UB}$ gives, after $K$ stages and an arbitrary feasible
continuation, an upper bound of
$V_{UB}(b^0)+\gamma^K(U_0-L_0)$. Letting $K\to\infty$ gives a feasible
strategy $\sigma_2^{UB}$ such that
\[
\sup_{\sigma_1}V(b^0;\sigma_1,\sigma_2^{UB})
\leq V_{UB}(b^0).
\]
Since the two bounds differ by at most $\epsilon$, neither player can gain
more than $\epsilon$ by deviating. Thus
$(\sigma_1^{LB},\sigma_2^{UB})$ is an $\epsilon$-Nash equilibrium,
equivalently an $\epsilon$-saddle point.

Finally, the game $\widehat G$ constructed in
\cref{thm:strategic-equivalence} satisfies the required condition. In an
original state, player~$1$ has the common feasible action set
$\{\bot_1\}$; in a dummy state, it has the common feasible action set
$A_1$. Player~$1$ observes which phase is being played, while player~$2$
observes the transformed state and uses $B(s)$ or $\{\bot_2\}$ as
appropriate. The result therefore applies to $\widehat G$.

\end{proofE}

The termination condition deteriorates as $\gamma\to1$. Under the
uniform reward bounds used here,
\[
\delta=\frac{\widehat R_{\max}-\widehat R_{\min}}{2(1-\sqrt{\gamma})},
\]
so the sufficient condition may be written as
\[
D<\frac{(1-\sqrt{\gamma})^2\epsilon}
{\widehat R_{\max}-\widehat R_{\min}}.
\]
This is a worst-case sufficient bound and becomes conservative for
discount factors close to one; we discuss its practical implications and our finite-horizon use of $\gamma=0.999$ when we discuss the limitations of our algorithm in \cref{subs:LimsofAlg}.

\subsection{Algorithm}
Using our representation of a history-dependent SA-MDP as a constrained zero-sum OS-POSG, we devise an algorithm to solve our game. The algorithm first converts a given SA-MDP game (as encoded in a custom text file format) into an equivalent constrained OS-POSG with phase-dependent singleton feasible action sets and adversarial proximity constraints. Given this constrained OS-POSG, we then leverage \citet{horak2023solving} with HSVI, modified to respect state-dependent feasibility constraints for both players. Further details of the algorithm are provided in \cref{app:algorithm}.

\section{Experiments}
\label{sec:experiments}
\subsection{Example SA-MDP Where History Matters}
\label{subs:example_hist_matters}
We construct an example history-dependent SA-MDP with analytically verifiable optimal strategies to explicitly find how much worse off a player facing a history-dependent adversary is and to provide a benchmark setting to test our algorithm.

The game consists of states $S = \{s_{0}, s_1, s_2, s_3\}$, actions $A_1=\{A,B,C\}$, and a time horizon $H=2.$ Adversaries can misreport only in restricted ways, where the restriction on the support of misreporting strategies is given by $B(s) \in 2^{S}\setminus \emptyset$ for state $s \in S$. For our problem $B(s_1) = \{s_1, s_2\}$, $B(s_2)=\{s_2,s_3\}$ and $B(s_3) = \{s_3, s_1\}$. Rewards are sparse and are only realized at the end of the game. The full game is shown in \cref{fig:tree}.
\begin{figure}[ht]
  \centering
  \begin{tikzpicture}[
    >=Stealth,
    scale=0.5,
    node distance=10mm,
    every node/.style={font=\small},
    agent/.style={circle,draw,thick,minimum size=4mm},
    chance/.style={diamond,draw,thick,aspect=1.6,inner sep=1pt},
    adv/.style={rectangle,draw,thick,rounded corners=2pt,minimum height=4mm,inner sep=2pt},
    terminal/.style={rectangle,minimum width=4mm,minimum height=3mm,inner sep=2pt}
  ]

    \node[agent] (root) {$s_{0}$};

    \node[chance, below left=9mm and 10mm of root] (cA) {};
    \node[above left=0mm of cA] {Nature};
    \draw[->] (root) -- node[above left] {$A$} (cA);

    \node[chance, below right=9mm and 10mm of root] (cB) {};
    \node[above right=0mm of cB] {Nature};
    \draw[->] (root) -- node[above right] {$B$} (cB);

    \node[adv, below=18mm of root, xshift=-26mm] (a3) {Adv $s_3$};
    \node[adv, below=18mm of root]              (a2) {Adv $s_2$};
    \node[adv, below=18mm of root, xshift=26mm] (a1) {Adv $s_1$};

    \draw[->] (cA) -- node[below left,pos=0.3] {$s_2,\; \tfrac12$} (a2);
    \draw[->] (cA) -- node[above left,pos=0.95]  {$s_3,\; \tfrac12$} (a3);

    \draw[->] (cB) -- node[above right,pos=0.95] {$s_1,\; \tfrac12$} (a1);
    \draw[->] (cB) -- node[below right,pos=0.3]  {$s_2,\; \tfrac12$} (a2);

    \node[agent, below=11mm of a3] (g3) {};
    \draw[->] (a3) -- node[left] {$\tilde s \in \{s_1,s_3\}$} (g3);

    \node[agent, below=11mm of a2] (g2) {};
    \draw[->] (a2) -- node[left] {$\tilde s \in \{s_2,s_3\}$} (g2);

    \node[agent, below=11mm of a1] (g1) {};
    \draw[->] (a1) -- node[right] {$\tilde s \in \{s_1,s_2\}$} (g1);

    \node[terminal, below left=13mm and 5mm of g1]  (t1A) {$0$};
    \node[terminal, below=13mm of g1]              (t1B) {$-2$};
    \node[terminal, below right=13mm and 5mm of g1] (t1C) {$0$};

    \draw[->] (g1) -- node[left]  {$A$} (t1A);
    \draw[->] (g1) -- node[right] {$B$} (t1B);
    \draw[->] (g1) -- node[right] {$C$} (t1C);

    \node[terminal, below left=13mm and 5mm of g2]  (t2A) {$-1$};
    \node[terminal, below=13mm of g2]              (t2B) {$2$};
    \node[terminal, below right=13mm and 5mm of g2] (t2C) {$-1$};

    \draw[->] (g2) -- node[left]  {$A$} (t2A);
    \draw[->] (g2) -- node[right] {$B$} (t2B);
    \draw[->] (g2) -- node[right] {$C$} (t2C);

    \node[terminal, below left=13mm and 5mm of g3]  (t3A) {$-2$};
    \node[terminal, below=13mm of g3]              (t3B) {$-2$};
    \node[terminal, below right=13mm and 5mm of g3] (t3C) {$0$};

    \draw[->] (g3) -- node[left]  {$A$} (t3A);
    \draw[->] (g3) -- node[right] {$B$} (t3B);
    \draw[->] (g3) -- node[right] {$C$} (t3C);

  \end{tikzpicture}
  \caption{Finite horizon version of SA-MDP game tree (action $C$ at $s_0$ yields $-10$ and so omitted as never chosen): the agent chooses $A$ or $B$ at
  $s_0$, Nature draws the second-stage state, the adversary reports a neighbor state $\tilde s\in B(s)$,
  and the agent chooses $A,B,$ or $C$ for the terminal payoff.}
  \label{fig:tree}
  \vspace{-5mm}
\end{figure}

An analysis of the game in \cref{app:example-samdp} demonstrates that the value associated with Markovian strategies for both the agent and the adversary $V_{\textrm{Markov}} = 0.25$, while for a history-dependent agent and adversary, we get $V_{\textrm{Hist}} = 0$. This means that the agent is strictly worse off facing a history-dependent adversary, even after being allowed history-dependent strategies. 

There are multiple equilibrium strategies for each player. Our code successfully recovers one of these equilibrium strategies that defines the boundary of the region containing optimal strategies under a barrier-point linear program solver for each stage game. Full recovered strategies are shown visually in \cref{fig:adv_markov_vs_hist_mapping} and written out explicitly in \cref{app:recovered-samdp}.

\tikzset{
  ballM/.style={
    circle, draw=black, thick, fill=NavyBlue!40,
    text=white, minimum size=5.2mm, inner sep=0pt, font=\scriptsize
  },
  ballH/.style={
    circle, draw=black, thick, fill=NavyBlue!40,
    text=white, minimum size=4.6mm, inner sep=0pt, font=\tiny
  },
  mapArrow/.style={->, thick, >=Latex, draw=black},
  prob/.style={font=\tiny, inner sep=0.4pt, text=black}
}

\newcommand{\MapOne}[2]{\draw[mapArrow] (#1) -- (#2);} 
\newcommand{\MapSplit}[5]{%
  \draw[mapArrow] (#1) -- node[pos=0.62, prob, above, yshift=1.2pt] {#3} (#2);
  \draw[mapArrow] (#1) -- node[pos=0.62, prob, above, yshift=1.2pt] {#5} (#4);
}

\begin{figure}[t]
\centering

\begin{minipage}[t]{0.48\columnwidth}
\vspace{0pt}\centering
\resizebox{0.6\linewidth}{!}{%
\begin{tikzpicture}
  \coordinate (C) at (0,0);

\node[font=\small\bfseries] at ($(C)+(0,10mm)$)
  {Markovian $\pi_2(\tilde s\mid s)$};

  \node[ballM] at ($(C)+(-18mm,0mm)$) (Ls1) {$s_1$};
  \node[ballM, below=4.6mm of Ls1]    (Ls2) {$s_2$};
  \node[ballM, below=4.6mm of Ls2]    (Ls3) {$s_3$};

  \node[ballM] at ($(C)+( 18mm,-4.6mm)$) (Rs2) {$s_2$};
  \node[ballM, below=4.6mm of Rs2]       (Rs3) {$s_3$};

  \MapOne{Ls1}{Rs2}
  \MapSplit{Ls2}{Rs2}{$\tfrac12$}{Rs3}{$\tfrac12$}
  \MapOne{Ls3}{Rs3}
\end{tikzpicture}%
}
\end{minipage}
\hfill
\begin{minipage}[t]{0.48\columnwidth}
\vspace{0pt}\centering
\resizebox{0.6\linewidth}{!}{%
\begin{tikzpicture}
  \coordinate (C) at (0,0);

\node[font=\small\bfseries] at ($(C)+(0,10mm)$)
  {History-dependent $\pi_2(\tilde s\mid s,a)$};

  \coordinate (CA) at ($(C)+(0,0mm)$);
  \node[font=\scriptsize] at ($(CA)+(0,6.8mm)$) {after $A$};

  \node[ballH] at ($(CA)+(-15mm,0mm)$) (ALs2) {$s_2$};
  \node[ballH, below=3.6mm of ALs2]    (ALs3) {$s_3$};

  \node[ballH] at ($(CA)+( 15mm,0mm)$) (ARs2) {$s_2$};
  \node[ballH, below=3.6mm of ARs2]    (ARs3) {$s_3$};

  \MapSplit{ALs2}{ARs2}{$\tfrac13$}{ARs3}{$\tfrac23$}
  \MapOne{ALs3}{ARs3}

\coordinate (CB) at ($(C |- ALs3.south) + (0,-9mm)$);

\node[font=\scriptsize] at ($(CB)+(0,7.2mm)$) {after $B$};

\node[ballH] at ($(CB)+(-15mm,0mm)$) (BLs1) {$s_1$};
\node[ballH, below=3.2mm of BLs1]    (BLs2) {$s_2$};

\node[ballH] at ($(CB)+( 15mm,-2.1mm)$) (BRs2) {$s_2$};

  \MapOne{BLs1}{BRs2}
  \MapOne{BLs2}{BRs2}
\end{tikzpicture}%
}
\end{minipage}

\caption{Adversary equilibrium strategies for Markovian vs.\ history-dependent adversary. Left: Markovian. Right: history-dependent kernels conditioned on player action in state $s_0$ ( $A$ vs.\ after $B$).}
\label{fig:adv_markov_vs_hist_mapping}
\vspace{-6mm}
\end{figure}

Our algorithm returns an upper and lower bound that contain the true value for the history-dependent strategies that are less than tolerance $\epsilon$ apart on convergence. The midpoint of these upper and lower bounds for history-dependent strategies is given by $\hat{V}_{\textrm{Hist}} = 0 \pm O(10^{-8})$, compared to a known analytic value of 0, for a convergence tolerance $\epsilon=0.01$. Because we use $\gamma \sim 1$ ($\gamma=1$ breaks HSVI), even with exact precision we would have small numerical error left over for an arbitrary problem. In practice, over moderate horizons, the difference is negligible. In comparison, our Markovian game solver finds a solution with a numerical tolerance of $O(10^{-6})$, based on the numerical precision of the linear programming solver. The final solution found via the linear solver is $\hat{V}_{\textrm{Markov}} = 0.2496 \pm O(10^{-6})$, compared to a true analytical value for $\gamma=0.999\sim 1$ of 0.2496 and $0.25$ for the exact $\gamma=1$. 
\subsection{Larger Problems}
\label{subs:largerProblems}
We test our algorithm up to depth 10 with 4 actions per node on randomly generated trees with 3.1 million states and 1.4 million information sets. Solving takes 2200s on a single MacBook M4 Pro CPU core and involves 18,000 explorations of the tree. We also tested our algorithm with 5 different seeds for games of sizes 2, 4, 6, and 8. We get a roughly log-linear plot in the amount of solve time per additional layer of depth. Roughly every 2 additional horizons increase solve time by 20--30$\times$. See \cref{fig:scaling_performance} and \cref{tab:depth_scaling} for detailed scaling information and \cref{app:algorithm} for further discussion.

\begin{table}[!htbp]
\centering
\caption{Runtime and exploration scaling by depth over 5 seeds. Explorations are the number required to reach an excess value-function gap below $10^{-2}$. Times are rounded to hundredths of a second.}
\label{tab:depth_scaling}
\begin{tabular}{@{}r r S[table-format=2.2] c r@{}}
\toprule
Depth & States & {Mean time (s)} & {Std. dev. (s)} & Explorations \\
\midrule
2 & 23     & 0.00  & 0.00  & 6 \\
4 & 743    & 0.08  & 0.02 & 28 \\
6 & 12,263 & 1.48  & 0.30 & 257 \\
8 & 196,583 & 52.49 & 21.49 & 2,564   \\
\bottomrule
\end{tabular}
\end{table}
Finally, we practically test our algorithm on Atari Freeway at depth 12 with a frame-skip of 30. Here, following the procedure in BRIDGE \citep{laidlaw2023} we generate an exact enumeration of state-actions reachable within the given horizon and merge behaviorally equivalent states. We found that on the default Atari settings with frameskip 30, agents were able to find an optimal deterministic strategy independent of actual observations on the initial reset, and thus immune to adversary perturbations. Therefore, we made the Atari initial state ``random'' by  hidden warm-up uncertainty and an initial adversary maximal perturbation of the radius in L2 distance as 500 (about 2.73 grayscale levels per pixel if spread uniformly), which permits multiple reports. The two initial states were generated by ``warming-up'' the Freeway environment and letting it run initially for a random length of time. We then searched across these possible random initial states until we found two whose distance was less than 500 apart from each other (their initial screen distance is 486.7018). This ensured that these two initial states are confusable. We then placed an equal probability of the agent starting in either of the two initial states in the Atari simulation.  Here, $V_{\textrm{HD},\textrm{HD}}$ lies below the no-adversary value and exceeds the approximate Markov FP value, showing history-dependent strategies matter and help the agent here. Atari Freeway is depth 12 but takes just 2.07s to solve for optimal history-dependent adversary and agent policies. This is because the number of info partitions in the OS-POSG is small after this behavioral consolidation (5,242 partitions versus 87,375 for depth 8, action 4 game).

In this Atari Freeway setting, history matters, and we find a history-dependent value of 0.8167225 (0.816584--0.816861 lower and upper value bounds from HSVI after converging under 0.01 excess gap; a reported excess gap of $\sim 2\times 10^{-4}$) for both the agent and the adversary exhibiting history-dependence, compared to 0.856748 for no adversary and 0.758971 for both agents being Markovian. Here, history dependence helps the agent detect manipulations by the adversary, leading to improved performance. Theoretically, however, the direction could go either way.

Larger settings are feasible to scale to, but are limited by the use of HSVI as our algorithm for exact solving; the primary constraint is memory as our depth 10 setting uses peak RAM of 8.6GB.

\subsection{Limitations of Algorithm}
\label{subs:LimsofAlg}
The primary limitation of our algorithm is that, for general games, the solver wall time scales in a roughly exponential fashion in the maximum length of history-dependence we solve for in the original SA-MDP problem. This is unsurprising, given that the number of histories themselves scale in this fashion and we are aiming to exactly solve the solution. 
In principle, however, any solution algorithm for general zero-sum OS-POSG games can now be applied. This means that approximate solutions involving depth-limited search, public belief-space methods, flexible neural network-based value function or policy function estimators learned via self-play could be applied to our algorithm, at the cost of our strong guarantees on convergence that we get from utilizing HSVI in exchange for even further scalability. Already, however, we have shown we can tackle SA-MDP problems of real-world relevance and size, while still retaining the use of HSVI with guarantees on the validity of the solution.
A further limitation arises from using a discount factor close to one to represent finite-horizon objectives within the infinite-horizon HSVI framework. In our finite-horizon experiments, we use $\gamma=0.999$ and
absorbing zero-reward states after the terminal depth. This introduces a small difference between the intended undiscounted finite-horizon return and the discounted return optimized by the solver. Thus, although $\gamma=0.999$ provides a close approximation to the undiscounted short-horizon objective, the corresponding worst-case termination bound can be highly conservative. Our experiments indicate that practical convergence is substantially better than this bound, but the approach remains most suitable for short- to moderate-horizon
problems.

 We leave it to future work to explore the application of these methods to solve larger SA-MDPs with our approach. 

\section{Conclusions}
We develop theory for history-dependent strategies in SA-MDPs, as well as an algorithm for computing their solutions using a novel reduction combined with HSVI. We show how this can be extended to study constrained adversaries. Finally, we verify in numerical examples that history-dependence quantitatively and qualitatively matters for optimal policies in SA-MDPs and how much an adversary can punish an agent and show that our algorithm successfully recovers Nash equilibrium strategies. This opens the door to scaling our approach to substantially larger state spaces by using compressed belief representations, using approximate solvers for POSG in a public belief space representation of the setting, and horizon truncation at the cost of guarantees and uniform bounds for HSVI.
\section*{Impact Statement} \label{sec:impact_statement}

This work advances the theoretical and algorithmic foundations of robust reinforcement learning under adversarial observation manipulation by formalizing and computing history-dependent equilibria in state-adversarial MDPs. By showing when history dependence fundamentally alters equilibrium behavior - and providing the first practical method to compute such equilibria - our results improve understanding of decision-making in safety-critical and adversarial environments. 

Potential positive impacts include more reliable autonomous systems and stronger robustness guarantees in sequential decision problems. As with many advances in adversarial modeling, these techniques could also be misused to design more effective attacks; however, we believe that making such vulnerabilities explicit is necessary for developing defenses and improving system resilience overall. Finally, the real-world applicability of our current approach is constrained by its bounded scalability.
\bibliographystyle{plainnat}
\bibliography{dissertation_bib}

@book{munkres2000topology,
  title     = {Topology},
  author    = {Munkres, James R.},
  year      = {2000},
  edition   = {2nd},
  publisher = {Prentice Hall, Inc.},
  address   = {Upper Saddle River, NJ},
  isbn      = {978-0131816299}
}

@incollection{kuhn1953extensive,
  author    = {Kuhn, Harold W.},
  title     = {Extensive Games and the Problem of Information},
  booktitle = {Contributions to the Theory of Games II},
  editor    = {Kuhn, Harold W. and Tucker, Albert W.},
  series    = {Annals of Mathematics Studies},
  volume    = {28},
  pages     = {193--216},
  publisher = {Princeton University Press},
  address   = {Princeton, NJ},
  year      = {1953}
}

@inproceedings{laidlaw2023,
  author    = {Laidlaw, Cassidy and Russell, Stuart J and Dragan, Anca},
  booktitle = {Advances in Neural Information Processing Systems},
  editor    = {A. Oh and T. Naumann and A. Globerson and K. Saenko and M. Hardt and S. Levine},
  pages     = {58953-59007},
  publisher = {Curran Associates, Inc.},
  title     = {Bridging RL Theory and Practice with the Effective Horizon},
  volume    = {36},
  year      = {2023}
}

@article{zhang2020robust,
  title={Robust deep reinforcement learning against adversarial perturbations on state observations},
  author={Zhang, Huan and Chen, Hongge and Xiao, Chaowei and Li, Bo and Liu, Mingyan and Boning, Duane and Hsieh, Cho-Jui},
  journal={Advances in Neural Information Processing Systems},
  volume={33},
  pages={21024--21037},
  year={2020}
}

@inproceedings{zhang2021robust,
  author    = {Huan Zhang and Hongge Chen and Duane S. Boning and Cho{-}Jui Hsieh},
  title     = {Robust Reinforcement Learning on State Observations with Learned Optimal Adversary},
  booktitle = {International Conference on Learning Representations, {ICLR}},
  year      = {2021},
  bibsource = {dblp computer science bibliography, https://dblp.org},

  note={arXiv:2101.08452},

}

@inproceedings{sun2021strongest,
  author       = {Yifan Wang and
                  Lukas Rahmann and
                  Olga Sorkine{-}Hornung},
  title        = {Geometry-Consistent Neural Shape Representation with Implicit Displacement
                  Fields},
  booktitle    = {International Conference on Learning Representations, {ICLR}},
  year         = {2022},
   note      = {arXiv:2106.05187}
  }

@article{BirmpasGHCRV21,
  author       = {Georgios Birmpas and
                  Jiarui Gan and
                  Alexandros Hollender and
                  Francisco J. Marmolejo Coss{\'{\i}}o and
                  Ninad Rajgopal and
                  Alexandros A. Voudouris},
  title        = {Optimally Deceiving a Learning Leader in Stackelberg Games},
  journal      = {J. Artif. Intell. Res.},
  volume       = {72},
  pages        = {507--531},
  year         = {2021}
}

@article{han2022solution,
  title     = {What is the Solution for State-Adversarial Multi-Agent Reinforcement Learning?},
  author    = {Songyang Han and Sanbao Su and Sihong He and Shuo Han and Haizhao Yang and Fei Miao},
  journal   = {Trans. Mach. Learn. Res.},
  year      = {2022},
  note = {arXiv:2212.02705},
  bibSource = {Semantic Scholar https://www.semanticscholar.org/paper/cfd2c668504c0a97e73fe6e40fe7fc869aa5a40a}
}

@inproceedings{franzmeyer2023illusory,
  title={Illusory attacks: Detectability matters in adversarial attacks on sequential decision-makers},
  author={Franzmeyer, Tim and McAleer, Stephen Marcus and Henriques, Joao F and Foerster, Jakob Nicolaus and Torr, Philip and Bibi, Adel and de Witt, Christian Schroeder},
  booktitle={The Second Workshop on New Frontiers in Adversarial Machine Learning},
  year={2023}
}

@article{horak2023solving,
  title={Solving zero-sum one-sided partially observable stochastic games},
  author={Hor{\'a}k, Karel and Bo{\v{s}}ansk{\`y}, Branislav and Kova{\v{r}}{\'\i}k, Vojt{\v{e}}ch and Kiekintveld, Christopher},
  journal={Artificial Intelligence},
  volume={316},
  pages={103838},
  year={2023},
  publisher={Elsevier}
}

@inproceedings{huang2017adversarial,
  title     = {Adversarial Attacks on Neural Network Policies},
  author    = {Sandy H. Huang and Nicolas Papernot and I. Goodfellow and Yan Duan and P. Abbeel},
  booktitle   = {International Conference on Learning Representations, {ICLR}},
  year      = {2017},
  note={arXiv:1702.02284}
}

@inproceedings{lin2017tactics,
author = {Lin, Yen-Chen and Hong, Zhang-Wei and Liao, Yuan-Hong and Shih, Meng-Li and Liu, Ming-Yu and Sun, Min},
title = {Tactics of adversarial attack on deep reinforcement learning agents},
year = {2017},
isbn = {9780999241103},
publisher = {AAAI Press},
booktitle = {Proceedings of the 26th International Joint Conference on Artificial Intelligence},
pages = {3756–3762},
numpages = {7},
location = {Melbourne, Australia},
series = {IJCAI'17}
}

@inproceedings{behzadan2017vulnerability,
  title={Vulnerability of deep reinforcement learning to policy induction attacks},
  author={Behzadan, Vahid and Munir, Arslan},
  booktitle={International conference on machine learning and data mining in pattern recognition},
  pages={262--275},
  year={2017},
  organization={Springer}
}

@article{kaelbling1998planning,
  title={Planning and acting in partially observable stochastic domains},
  author={Kaelbling, Leslie Pack and Littman, Michael L and Cassandra, Anthony R},
  journal={Artificial intelligence},
  volume={101},
  number={1-2},
  pages={99--134},
  year={1998},
  publisher={Elsevier}
}

@inproceedings{smith2012heuristic,
author = {Smith, Trey and Simmons, Reid},
title = {Heuristic search value iteration for POMDPs},
year = {2004},
isbn = {0974903906},
publisher = {AUAI Press},
address = {Arlington, Virginia, USA},
booktitle = {Proceedings of the 20th Conference on Uncertainty in Artificial Intelligence},
pages = {520–527},
numpages = {8},
location = {Banff, Canada},
series = {UAI '04}
}

@booklet{cplex-manual,
  author  = {{IBM}},
  title   = {User's Manual for CPLEX},
  version = {22.1.2},
  url     = {https://www.ibm.com/docs/en/icos/22.1.2?topic=optimizers-users-manual-cplex},
  urldate = {2026-01-27},
  year={2024},
  note    = {IBM ILOG CPLEX Optimization Studio documentation},
}

@article{oikarinen2021robust,
  title={Robust deep reinforcement learning through adversarial loss},
  author={Oikarinen, Tuomas and Zhang, Wang and Megretski, Alexandre and Daniel, Luca and Weng, Tsui-Wei},
  journal={Advances in Neural Information Processing Systems},
  volume={34},
  pages={26156--26167},
  year={2021}
}

@inproceedings{liang2022efficient,
  title     = {Efficient Adversarial Training without Attacking: Worst-Case-Aware Robust Reinforcement Learning},
  author    = {Liang, Yongyuan and Sun, Yanchao and Zheng, Ruijie and Huang, Furong},
  booktitle = {Advances in Neural Information Processing Systems},
  volume    = {35},
  year      = {2022}
}

@inproceedings{mcmahan2023optimal,
  title     = {Optimal Attack and Defense for Reinforcement Learning},
  author    = {McMahan, Jeremy and Wu, Young and Zhu, Xiaojin and Xie, Qiaomin},
  booktitle = {Proceedings of the AAAI Conference on Artificial Intelligence},
  year      = {2024},
  note      = {arXiv:2312.00198}
}

@inproceedings{liang2023game,
  title     = {Game-Theoretic Robust Reinforcement Learning Handles Temporally-Coupled Perturbations},
  author    = {Liang, Yongyuan and Sun, Yanchao and Zheng, Ruijie and Liu, Xiangyu and Eysenbach, Benjamin and Sandholm, Tuomas and Huang, Furong and McAleer, Stephen},
  booktitle = {International Conference on Learning Representations (ICLR)},
  year      = {2024},
  note      = {arXiv:2307.12062}
}

@article{sarkadi2019modelling,
  title={Modelling deception using theory of mind in multi-agent systems},
  author={Sarkadi, {\c{S}}tefan and Panisson, Alison R and Bordini, Rafael H and McBurney, Peter and Parsons, Simon and Chapman, Martin},
  journal={AI Communications},
  volume={32},
  number={4},
  pages={287--302},
  year={2019},
  publisher={SAGE Publications Sage UK: London, England}
}

@article{motwani2024secret,
  title={Secret collusion among ai agents: Multi-agent deception via steganography},
  author={Motwani, Sumeet R and Baranchuk, Mikhail and Strohmeier, Martin and Bolina, Vijay and Torr, Philip H and Hammond, Lewis and de Witt, Christian S},
  journal={Advances in Neural Information Processing Systems},
  volume={37},
  pages={73439--73486},
  year={2024}
}

@article{wang2022novel,
  title={A novel bipartite consensus tracking control for multiagent systems under sensor deception attacks},
  author={Wang, Xinjun and Cao, Ye and Niu, Ben and Song, Yongduan},
  journal={IEEE Transactions on Cybernetics},
  volume={53},
  number={9},
  pages={5984--5993},
  year={2022},
  publisher={IEEE}
}

@inproceedings{tu2021adversarial,
  title={Adversarial attacks on multi-agent communication},
  author={Tu, James and Wang, Tsunhsuan and Wang, Jingkang and Manivasagam, Sivabalan and Ren, Mengye and Urtasun, Raquel},
  booktitle={Proceedings of the IEEE/CVF International Conference on Computer Vision},
  pages={7768--7777},
  year={2021}
}

@article{pendharkar2012game,
  title={Game theoretical applications for multi-agent systems},
  author={Pendharkar, Parag C},
  journal={Expert Systems with Applications},
  volume={39},
  number={1},
  pages={273--279},
  year={2012},
  publisher={Elsevier}
}

@article{panait2005cooperative,
  title={Cooperative multi-agent learning: The state of the art},
  author={Panait, Liviu and Luke, Sean},
  journal={Autonomous agents and multi-agent systems},
  volume={11},
  number={3},
  pages={387--434},
  year={2005},
  publisher={Springer}
}

@book{wang2017cooperative,
  title={Cooperative control of multi-agent systems: Theory and applications},
  author={Wang, Yue and Garcia, Eloy and Casbeer, David and Zhang, Fumin},
  year={2017},
  publisher={John Wiley \& Sons}
}

@article{cai2022cooperative,
  title={Cooperative control of multi-agent systems},
  author={Cai, He and Su, Youfeng and Huang, Jie},
  journal={Cham: Springer-Verlag},
  year={2022},
  publisher={Springer}
}

@article{wang2022cooperative,
  title={Cooperative and competitive multi-agent systems: From optimization to games},
  author={Wang, Jianrui and Hong, Yitian and Wang, Jiali and Xu, Jiapeng and Tang, Yang and Han, Qing-Long and Kurths, J{\"u}rgen},
  journal={IEEE/CAA Journal of Automatica Sinica},
  volume={9},
  number={5},
  pages={763--783},
  year={2022},
  publisher={IEEE}
}

@InProceedings{khakpour2022partially,
author="Khakpour, Narges
and Parker, David",
editor="Madeira, Alexandre
and Knapp, Alexander",
title="Partially-Observable Security Games for Attack-Defence Analysis in Software Systems",
booktitle="Software Engineering and Formal Methods",
year="2025",
publisher="Springer Nature Switzerland",
address="Cham",
pages="144--161",
isbn="978-3-031-77382-2"
}
\newpage
\appendix
\section{Proofs}
\label{app:deferred-proofs}

\printProofs


\section{Numerically Discovered SA-MDP Strategies}
\label{app:recovered-samdp}
We use the standard game theory notation for mixed strategies $p[A] +(1-p)[B]$ to indicate a mixed strategy that plays action $A$ with probability $p$ and $B$ with probability $1-p$, and $A$ to indicate the pure strategy that plays $A$. We write the adversary's reporting strategies the same way.

\subsection{Markovian Strategies}
\paragraph{Agent} We find Markovian equilibrium strategies,
$\pi_1(s_0) = A,\ \pi_1(s_1) = A,\ \pi_1(s_2)= B,\ \pi_1(s_3)= C$, given observations $s_0, s_1, s_2, s_3$.
\paragraph{Adversary} We find equilibrium strategies 
$\pi_2(\cdot|s_0)=s_0, \pi_2(\cdot |s_1)= s_2, \pi_2(\cdot|s_2)= \tfrac{1}{2}[s_2]+\tfrac{1}{2}[s_3]$, $\pi_2(\cdot|s_3)=s_3$.
\subsection{History Dependent Strategies }
\paragraph{Agent} Let $\pi_1(s_0)$ denote our agent's strategy at state $s_0$ in the first stage, and $\pi_1(a,\tilde{s})$ denote our agent's optimal strategy after playing action $a$ in the first stage and receiving observation $\tilde{s}$. Our agent's strategy is $\pi_1(s_0)=A$, $\pi_1(A, \tilde{s}_i) = B$ for $\tilde{s}_i \in \{s_2, s_3\}$ for on-policies strategies. Off-policy, our agent best responds $\pi_1(B, \tilde{s}_i) = B$ for $\tilde{s}_i \in \{s_2,s_3\}$.
Finally, again off-policy, our agent best responds $\pi_1(C, s_0)=A.$

\paragraph{Adversary} Our adversary best responds as:
At state $s_0$: $\pi_2(\cdot| s_0) = s_0$.
After $A$ (on policy): $\pi_2(\cdot| s_2, A) = \tfrac{1}{3}[s_2]+\tfrac{2}{3}[s_3]$, $\pi_2(\cdot|s_3, A) = s_3$.
After $B$, off-policy: $\pi_2(\cdot|s_i, B) = s_2$, for $s_i \in \{s_1, s_2\}$.
After $C$, off-policy: $\pi_2(\cdot|s_0, C) = s_0$ (which is forced due to trivial constraint).
\section{Analysis of Example SA-MDP Game Where History Matters}\label{app:example-samdp}
\subsection{Setup}
The state space and initial distribution are given by
\(
    S=\{s_0,s_1,s_2,s_3, t\},\ b^0=\delta_{s_0}.
\)
The action space is \(
  A_1=\{A,B,C\}.
\)
\citet{horak2023solving} solve for infinite horizon settings. Therefore, we  choose $\gamma \sim 1$, and add absorbing states after $t=2$ that return reward 0 from then on to convert our finite horizon setting to an infinite horizon game. Similarly, $t$ is added to our state space. To simplify the notation of our SA-MDP, we drop this dependence on $t$, with the understanding that after $t=2$, rewards are 0 from then on and all transition probabilities from a state to itself are 1 regardless of the action taken.
\paragraph{Transitions}
Only transitions out of $s_0$ are nontrivial. Transitions are given by
\begin{align}
  p(s_2\mid s_0,A) &= \tfrac12, & p(s_3\mid s_0,A) &= \tfrac12, & p(s_1\mid s_0,A) &= 0, \label{eq:prior-A}\\
  p(s_1\mid s_0,B) &= \tfrac12, & p(s_2\mid s_0,B) &= \tfrac12, & p(s_3\mid s_0,B) &= 0, \label{eq:prior-B}\\
  p(s_0\mid s_0,C) &= 1. &&&& \label{eq:prior-C}
\end{align}
All other states are absorbing.

\paragraph{Reward}
The reward $R(s,a)$ is
\begin{equation*}
  \label{eq:reward-table}
  \begin{array}{c|ccc}
      & A & B & C \\
    \hline
    s_0 & 0 & 0 & -10 \\
    s_1 & 0 & -2 & 0 \\
    s_2 & -1 & 2 & -1 \\
    s_3 & -2 & -2 & 0
  \end{array}
  \end{equation*}

\paragraph{Adversary Constraint Sets}
We restrict the adversary to report only ``neighboring'' states:
\begin{equation}
  \label{eq:neighbor}
  \begin{aligned}
    B(s_0) &= \{s_0\},\\
    B(s_1) &= \{s_1,s_2\},\\
    B(s_2) &= \{s_2,s_3\},\\
    B(s_3) &= \{s_1,s_3\}.
  \end{aligned}
\end{equation}
where at $s_0$ the adversary has to report the (trivial) true initial state.

\subsection{Markovian vs history-dependent observation adversaries}

A \emph{Markovian} adversary is any kernel
$\pi_2(\tilde s\mid s)$ such that $\mathrm{supp}(\pi_2(\cdot\mid s))\subseteq B(s)$ for each $s$.
Because $B(s_0)=\{s_0\}$, we have $\pi_2(s_0\mid s_0)=1$.

For $s_1,s_2,s_3$, every such Markovian kernel can be parameterized by $(q_1,q_2,q_3)\in[0,1]^3$:
\begin{equation}
  \label{eq:nu-param-markov}
  \begin{aligned}
    \pi_2(s_1\mid s_1) &= q_1,\quad &\pi_2(s_2\mid s_1) &= 1 - q_1, \\
    \pi_2(s_2\mid s_2) &= q_2,\quad &\pi_2(s_3\mid s_2) &= 1 - q_2, \\
    \pi_2(s_1\mid s_3) &= 1 - q_3,\quad &\pi_2(s_3\mid s_3) &= q_3.
  \end{aligned}
\end{equation}

A \emph{history-dependent} adversary is allowed to choose its reporting kernel as a function of the full history (past states/actions/observations).
So the adversary is powerless at $s_0$ (it must report $s_0$), and only matters at the second stage.

\subsection{Markovian vs history-dependent observation adversaries}

A \emph{Markovian} adversary is given by a ``Markov misreporting kernel''
$\pi_2(\tilde s\mid s)$ such that $\mathrm{supp}(\pi_2(\cdot\mid s))\subseteq B(s)$ for each $s$.
Because $B(s_0)=\{s_0\}$, we have $\pi_2(s_0\mid s_0)=1$.

For $s_1,s_2,s_3$, every such Markovian kernel can be parameterized by $(q_1,q_2,q_3)\in[0,1]^3$:
\begin{equation}
  \label{eq:nu-param-hd}
  \begin{aligned}
    \pi_2(s_1\mid s_1) &= q_1,\quad &\pi_2(s_2\mid s_1) &= 1 - q_1, \\
    \pi_2(s_2\mid s_2) &= q_2,\quad &\pi_2(s_3\mid s_2) &= 1 - q_2, \\
    \pi_2(s_1\mid s_3) &= 1 - q_3,\quad &\pi_2(s_3\mid s_3) &= q_3.
  \end{aligned}
\end{equation}

A \emph{history-dependent} adversary is allowed to choose its reporting kernel as a function of the
full history (past states/actions/observations). 
\subsection{Markovian adversary: best responses and value}
The analysis of the Markovian adversary is a fair amount of exhaustion by cases due to the division of the simplex into 4 regions, each satisfying both, either, or none of them.

\paragraph{Setup}
Fix a Markovian misreporting kernel \eqref{eq:nu-param-hd}.  Let $v_A(q_2,q_3)$ denote the agent's
best-response value if it commits to playing action $A$ in state $s_0$, and similarly
$v_B(q_1,q_2)$ the value if it commits to playing action $B$ in $s_0$.
Against the kernel $(q_1,q_2,q_3)$, the agent's best value is
\[
  F(q_1,q_2,q_3) \coloneqq \max\{ v_A(q_2,q_3),\; v_B(q_1,q_2)\}.
\]
We compute $v_A$ and $v_B$, then minimize $F$ over $(q_1,q_2,q_3)\in[0,1]^3$.

\paragraph{Branch $A$: value $v_A(q_2,q_3)$}

Under action $A$ at $s_0$, the prior at the second stage is $p^A=(0,\tfrac12,\tfrac12)$ over
$(s_1,s_2,s_3)$.

From $s_2$ and $s_3$, the observation distributions are
\begin{itemize}
  \item $s_2$: report $s_2$ with prob.\ $q_2$, report $s_3$ with prob.\ $1-q_2$.
  \item $s_3$: report $s_1$ with prob.\ $1-q_3$, report $s_3$ with prob.\ $q_3$.
\end{itemize}

Total report probabilities under the $A$-branch:
\[
\begin{aligned}
  \Prob(\tilde s=s_1\mid A) &= 0.5(1-q_3), \\
  \Prob(\tilde s=s_2\mid A) &= 0.5 q_2, \\
  \Prob(\tilde s=s_3\mid A) &= 0.5(1-q_2) + 0.5 q_3.
\end{aligned}
\]

\paragraph{Report $s_1$.}
Only $s_3$ can produce report $s_1$, so $P(s_3\mid \tilde s=s_1)=1$. From
\eqref{eq:reward-table}, action $C$ is uniquely optimal at $s_3$ with value $0$.

\paragraph{Report $s_2$.}
Only $s_2$ can produce report $s_2$, so $P(s_2\mid \tilde s=s_2)=1$. At $s_2$, action $B$ is
optimal with value $2$.

\paragraph{Report $s_3$.}
Report $s_3$ can come from $s_2$ or $s_3$. Using Bayes' rule to update beliefs, given
\[
\begin{aligned}
  w_2 &\coloneqq \Prob(\tilde s=s_3, s_2\mid A) = 0.5(1-q_2), \\
  w_3 &\coloneqq \Prob(\tilde s=s_3, s_3\mid A) = 0.5 q_3
  \end{aligned}\]
we get
\[ 
\begin{aligned}
P(s_2\mid \tilde s=s_3,A) &= \frac{w_2}{w_2+w_3}
  = \frac{1-q_2}{1-q_2+q_3} =: p.
\end{aligned}
\]
Expected payoffs for the report $s_3$:
\[
\begin{aligned}
  \mathbb{E}[R\mid A,\tilde s=s_3,A] &= p(-1) + (1-p)(-2) = -2 + p,\\
  \mathbb{E}[R\mid B,\tilde s=s_3,A] &= p(2) + (1-p)(-2) = -2 + 4p,\\
  \mathbb{E}[R\mid C,\tilde s=s_3,A] &= p(-1) + (1-p)(0) = -p.
\end{aligned}
\]
Comparisons show that
\begin{itemize}
  \item $B$ beats $C$ iff $p\ge 0.4$.
  \item $B$ always beats $A$ (for $p\ge0$).
  \item $C$ always beats $A$ (for $p\le1$).
\end{itemize}
So, the best action for the report $s_3$ is $B$ if $p\ge0.4$ and $C$ if $p<0.4$.
In terms of $(q_2,q_3)$,
\[
  p\ge0.4
  \iff \frac{1-q_2}{1-q_2+q_3}\ge0.4
  \iff 3 - 3q_2 - 2q_3 \ge 0.
\]
Define the region
\[
  C_A \coloneqq \{(q_2,q_3): 3 - 3q_2 - 2q_3 \ge 0\}.
\]

\paragraph{Branch-$A$ value.}
\emph{Case A1: $(q_2,q_3)\in C_A$ (so report $s_3\mapsto B$).}

From $s_2$ we always play $B$ and get $2$. From $s_3$ we play $C$ with prob.\ $1-q_3$ (payoff $0$)
and $B$ with prob.\ $q_3$ (payoff $-2$), so $\mathbb{E}[R\mid s_3] = -2q_3$.
Thus,
\[
  v_A(q_2,q_3) = 0.5\cdot 2 + 0.5\cdot (-2q_3) = 1 - q_3.
\]

\emph{Case A2: $(q_2,q_3)\notin C_A$ (so report $s_3\mapsto C$).}

From $s_2$, report $s_2$ leads to $B$ (payoff $2$) and report $s_3$ leads to $C$ (payoff $-1$),
so
\[
  \mathbb{E}[R\mid s_2] = q_2\cdot 2 + (1-q_2)(-1) = 3q_2 - 1.
\]
From $s_3$ we always play $C$ and get $0$. Hence,
\[
  v_A(q_2,q_3) = 0.5(3q_2-1)+0.5(0) = 1.5q_2 - 0.5.
\]

In summary:
\begin{equation}
  \label{eq:vA}
  v_A(q_2,q_3)=
  \begin{cases}
    1 - q_3, & \text{if } 3 - 3q_2 - 2q_3 \ge 0,\\
    1.5q_2 - 0.5, & \text{if } 3 - 3q_2 - 2q_3 < 0.
  \end{cases}
\end{equation}

\paragraph{Branch $B$: value $v_B(q_1,q_2)$}

Under action $B$ at $s_0$, the prior at the second stage is $p^B=(\tfrac12,\tfrac12,0)$.

From $s_1$ and $s_2$:
\begin{itemize}
  \item $s_1$: report $s_1$ with prob.\ $q_1$, report $s_2$ with prob.\ $1-q_1$.
  \item $s_2$: report $s_2$ with prob.\ $q_2$, report $s_3$ with prob.\ $1-q_2$.
\end{itemize}

Total report probabilities under the $B$-branch:
\[
\begin{aligned}
  \Prob(\tilde s=s_1\mid B) &= 0.5q_1,\\
  \Prob(\tilde s=s_2\mid B) &= 0.5(1-q_1)+0.5q_2,\\
  \Prob(\tilde s=s_3\mid B) &= 0.5(1-q_2).
\end{aligned}
\]

\paragraph{Report $s_1$.}
Only $s_1$ can produce report $s_1$; best action is $A$ (value $0$).

\paragraph{Report $s_3$.}
Only $s_2$ can produce report $s_3$; best action is $B$ (value $2$).

\paragraph{Report $s_2$.}
Report $s_2$ can come from $s_1$ or $s_2$. Let
\[
\begin{aligned}
  w_1 &= \Prob(\tilde s=s_2,s_1\mid B) = 0.5(1-q_1),\\
  w_2 &= \Prob(\tilde s=s_2,s_2\mid B) = 0.5q_2.
\end{aligned}
\]
Then
\begin{align*}
  P(s_1\mid \tilde s=s_2,B)&=\frac{1-q_1}{1-q_1+q_2}\\
  P(s_2\mid \tilde s=s_2,B)&=\frac{q_2}{1-q_1+q_2}.
\end{align*}
Expected payoffs for the report $s_2$:
\[
\begin{aligned}
  \mathbb{E}[R\mid A,\tilde s=s_2,B] &= -\frac{q_2}{1-q_1+q_2},\\
  \mathbb{E}[R\mid C,\tilde s=s_2,B] &= -\frac{q_2}{1-q_1+q_2},\\
  \mathbb{E}[R\mid B,\tilde s=s_2,B] &= \frac{-2(1-q_1)+2q_2}{1-q_1+q_2}
  = \frac{-2+2q_1+2q_2}{1-q_1+q_2}.
\end{aligned}
\]
So $B$ beats $A$/$C$ iff
\[
  -2+2q_1+2q_2 \ge -q_2
  \iff 2q_1+3q_2 \ge 2.
\]
Define the region
\[
  C_B \coloneqq \{(q_1,q_2): 2q_1+3q_2 \ge 2\}.
\]

\paragraph{Branch-$B$ value.}
\emph{Case B1: $(q_1,q_2)\in C_B$ (so report $s_2\mapsto B$).}

Mapping: report $s_1\mapsto A$, report $s_2\mapsto B$, report $s_3\mapsto B$.

From $s_1$ we know that report $s_1$ gives $A$ (payoff $0$) and report $s_2$ gives $B$ (payoff $-2$), so
$\mathbb{E}[R\mid s_1]=-2(1-q_1)$.
From $s_2$: both reports lead to $B$, so $\mathbb{E}[R\mid s_2]=2$.
Thus,
\[
  v_B(q_1,q_2)=0.5[-2(1-q_1)]+0.5\cdot 2 = q_1.
\]

\emph{Case B2: $(q_1,q_2)\notin C_B$ (so report $s_2\mapsto A$).}

Mapping: report $s_1\mapsto A$, report $s_2\mapsto A$, report $s_3\mapsto B$.

From $s_1$: always $A$, so $\mathbb{E}[R\mid s_1]=0$.
From $s_2$: report $s_2$ gives $A$ ($-1$), report $s_3$ gives $B$ ($2$), so
$\mathbb{E}[R\mid s_2]=q_2(-1)+(1-q_2)2 = 2-3q_2$.
Hence,
\[
  v_B(q_1,q_2)=0.5(0)+0.5(2-3q_2)=1-1.5q_2.
\]

In summary:
\begin{equation}
  \label{eq:vB}
  v_B(q_1,q_2)=
  \begin{cases}
    q_1, & \text{if } 2q_1+3q_2 \ge 2,\\
    1-1.5q_2, & \text{if } 2q_1+3q_2 < 2.
  \end{cases}
\end{equation}

\subsection{Minimizing the Markovian worst-case value}

For each kernel $(q_1,q_2,q_3)$, the agent's best-response value is
\[
  F(q_1,q_2,q_3) = \max\{ v_A(q_2,q_3),\ v_B(q_1,q_2)\},
\]
with $v_A$ and $v_B$ given by \eqref{eq:vA}--\eqref{eq:vB}. The Markovian
adversary's minimax problem is
\[
  V_{\text{Markov}} \coloneqq \inf_{q_1,q_2,q_3\in[0,1]} F(q_1,q_2,q_3).
\]

The domain splits into four regions depending on whether the conditions
defining $v_A$ and $v_B$ hold:
\[
  \begin{aligned}
    C_A &: 3 - 3 q_2 - 2 q_3 \ge 0,\\
    C_B &: 2 q_1 + 3 q_2 \ge 2.
  \end{aligned}
\]

We consider all four combinations. On each region, $v_A$ and $v_B$ are
linear, and $F$ is the maximum of two linear functions.

\paragraph{Region 1: $C_A$ and $C_B$ both hold.}

Here $v_A = 1 - q_3$ and $v_B = q_1$, so
\[
  F = \max\{1 - q_3,\ q_1\}.
\]
To minimize $F$ we equalize the two arguments
\[
q_1 = 1 - q_3. \label{eq:mint1t3}
\]
We substitute these into $C_B$
\[
  2(1-q_3) + 3 q_2 \ge 2
  \iff 3 q_2 \ge 2 q_3
  \iff q_2 \ge \tfrac{2}{3} q_3.
\]
From $C_A$
\[
  3 - 3 q_2 - 2 q_3 \ge 0
  \iff 3 q_2 \le 3 - 2 q_3
  \iff q_2 \le 1 - \tfrac{2}{3} q_3.
\]
Together, both constraints imply
\[
  \tfrac{2}{3} q_3 \le 1 - \tfrac{2}{3} q_3
  \iff q_3 \le \tfrac{3}{4}.
\]
On the equality line \cref{eq:mint1t3}, $F = q_1 = 1 - q_3$, which is minimized
by taking $q_3$ as large as possible, i.e.\ $q_3 = 3/4$, giving
\[
  F_{\min,\text{Region 1}} = 1 - \tfrac{3}{4} = \tfrac{1}{4}.
\]

\paragraph{Region 2: $C_A$ holds, $C_B$ fails.}

Here $v_A = 1 - q_3$ and $v_B = 1 - 1.5 q_2$, so
\[
  F = \max\{1 - q_3,\ 1 - 1.5 q_2\}.
\]
Equalizing gives
\[
  1 - q_3 = 1 - 1.5 q_2 \iff q_3 = 1.5 q_2.
\]

Constraint $C_A$ with $q_3 = 1.5 q_2$:
\[
  3 - 3 q_2 - 2(1.5 q_2) \ge 0
  \iff 3 - 6 q_2 \ge 0
  \iff q_2 \le 0.5.
\]
The condition $\lnot C_B: 2 q_1 + 3 q_2 < 2$ only restricts $q_1$ and can be
satisfied for some $q_1$ at any given $q_2\le 0.5$, so it does not constrain $q_2$ further.

On the equality line $q_3 = 1.5 q_2$, the value is
\[
  F = 1 - 1.5 q_2,
\]
minimized by taking $q_2$ as large as possible, i.e.\ $q_2 = 0.5$, giving
\[
  F_{\min,\text{Region 2}} = 1 - 1.5\cdot 0.5 = 0.25.
\]

\paragraph{Region 3: $C_A$ fails, $C_B$ holds.}

Here $v_A = 1.5 q_2 - 0.5$ and $v_B = q_1$, so
\[
  F = \max\{1.5 q_2 - 0.5,\ q_1\}.
\]
Equalizing gives
\[
  q_1 = 1.5 q_2 - 0.5.
\]
Constraint $C_B$ becomes
\begin{align*}
  2(1.5 q_2 - 0.5) + 3 q_2 \ge 2
  &\iff 3 q_2 - 1 + 3 q_2 \ge 2\\
  &\iff 6 q_2 \ge 3\\
  &\iff q_2 \ge 0.5.
\end{align*}
The condition $\lnot C_A: 3 - 3 q_2 - 2 q_3 < 0$ is
\[
  3 q_2 + 2 q_3 > 3.
\]
For any $q_2\ge 0.5$ this can be satisfied (e.g., \ by choosing $q_3$ close to
$1$), so it does not restrict $q_2$ beyond $q_2\ge 0.5$.

On the equality line, $F = q_1 = 1.5 q_2 - 0.5$, minimized at $q_2=0.5$:
\[
  F_{\min,\text{Region 3}} = 1.5\cdot 0.5 - 0.5 = 0.25.
\]

\paragraph{Region 4: $C_A$ and $C_B$ both fail.}

Here $v_A = 1.5 q_2 - 0.5$ and $v_B = 1 - 1.5 q_2$, so
\[
  F = \max\{1.5 q_2 - 0.5,\ 1 - 1.5 q_2\},
\]
which depends only on $q_2$.

Equalizing:
\[
  1.5 q_2 - 0.5 = 1 - 1.5 q_2
  \iff 3 q_2 = 1.5
  \iff q_2 = 0.5.
\]
At $q_2=0.5$ we get
\[
  F = 1.5\cdot 0.5 - 0.5 = 0.25.
\]
Here, the failure of $C_A$ and $C_B$ requires
\begin{align*}
  3 - 3 q_2 - 2 q_3 < 0 &\iff q_3 > 0.75,\\
  2 q_1 + 3 q_2 < 2 &\iff q_1 < 0.25,
\end{align*}
so, for example, $(q_1,q_2,q_3)=(0,0.5,0.8)$ lies in Region 4 and has $F=0.25$.
\paragraph{Conclusion}

In all four regions, the infimum of $F(q_1,q_2,q_3)$ is $1/4$. We also have
explicit triples (for example $(0,0.5,0.8)$) that achieve $F=1/4$. Thus,
\begin{align*}
  V_{\text{Markov}} &= \inf_{q_1,q_2,q_3\in[0,1]} F(q_1,q_2,q_3) \\
  &= \min_{q_1, q_2, q_3 \in [0,1]} F(q_1,q_2,q_3)\\
  &= \frac{1}{4}.
\end{align*}

Moreover, for any such minimizing kernel $\pi_2^*$, the best-response value is
exactly $1/4$, so \emph{no (pure) agent strategy can achieve expected reward
strictly larger than $1/4$} against this optimal Markovian adversary.

\subsection{History-dependent adversary}

Now, allow a history-dependent adversary that can use one kernel on the $A$-subtree and a different
kernel on the $B$-subtree. Concretely, at the second stage it may choose:
\[
  \pi_2^A(\cdot\mid s)\ \text{after first-stage action }A\] and \[
  \pi_2^B(\cdot\mid s)\ \text{after first-stage action }B,
\]
each respecting the same neighbor constraint \eqref{eq:neighbor}.

Then, two crucial facts in this new history dependent strategy are:
\begin{enumerate}
  \item The agent has a baseline strategy with value $0$ against \emph{any} adversary: play any mix of
  $A$ and $B$ at $s_0$, and at the second stage always play action $B$ regardless of the reported
  state. Under the $A$-branch, states $\{s_2,s_3\}$ yield payoffs $+2$ and $-2$ with equal
  probability; under the $B$-branch, states $\{s_1,s_2\}$ yield payoffs $-2$ and $+2$ with equal
  probability. In both cases, the expectation is $0$, and misreporting does not matter if reports are
  ignored. Hence, $V_{\text{Hist}}\ge 0$.
  \item On each branch separately, the adversary can drive the agent's best-response value down to $0$.
\end{enumerate}
The first fact is self-explanatory; however, the second fact requires some additional work. For the second fact, we show that there are branch-specific values that satisfy
\[
  \inf_{q_2,q_3} v_A(q_2,q_3)=0,\qquad \inf_{q_1,q_2} v_B(q_1,q_2)=0.
\]
For example,
\begin{itemize}
  \item On the $A$-branch, take $(q_2,q_3)=(0,1)$, which lies in $C_A$ and gives $v_A=1-q_3=0$.
  \item On the $B$-branch, take $(q_1,q_2)=(0,2/3)$, which lies in $\lnot C_B$ and gives
  $v_B=1-1.5\cdot(2/3)=0$.
\end{itemize}

A history-dependent adversary can therefore use $\pi_2^A$ with $(q_2,q_3)=(0,1)$ on the $A$-subtree
and $\pi_2^B$ with $(q_1,q_2)=(0,2/3)$ on the $B$-subtree. Then any policy that commits to $A$ can
achieve at most $0$ on that branch, any policy that commits to $B$ can achieve at most $0$ on that
branch, and mixing yields a convex combination, still at most $0$.
Therefore, $V_{\text{Hist}} \leq  0$. But by the fact that there is the baseline where the agent just ignores misreports and still obtains 0, we know that $V_{\text{Hist}} \geq 0$

Therefore, combining the baseline strategy with the examples we found, we obtain the exact value
\[
  V_{\text{Hist}} = 0.
\]
\subsection{Conclusion of $V_{\textrm{Markov}}$ and $V_{\textrm{Hist}}$ Comparison}
Thus, the history-dependent observation adversary is strictly more damaging
\[
  V_{\text{Hist}}=0 \;<\; \frac14 = V_{\text{Markov}}.
\]
\section{Algorithm Details}
\label{app:algorithm}
A high-level overview of our algorithm is provided below in \cref{alg:samdp_pipeline} and \cref{alg:constrained_hsvi}.

In \cref{alg:samdp_pipeline,alg:constrained_hsvi}, we write $V^{\Gamma}_{LB}$ and $V^{\Upsilon}_{UB}$ for the lower- and upper-bound value functions denoted above by $V_{LB}$ and $V_{UB}$. 
The superscripts emphasize that the lower bound is represented by a set $\Gamma$ of $\alpha$-vectors, whereas the upper bound is constructed from a set $\Upsilon$ of belief--value points. 
Thus, $V^{\Gamma}_{LB}$ and $V^{\Upsilon}_{UB}$ refer to the same bound functions as $V_{LB}$ and $V_{UB}$, with their finite representations made explicit.

\begin{algorithm}[bth]
\caption{Compute $\epsilon$-equilibrium for history-dependent SA-MDP}
\label{alg:samdp_pipeline}
\begin{algorithmic}[1]
\REQUIRE SA-MDP $M=(S,A_1,R,p,\gamma,B)$, initial belief $b^0\in\Delta(S)$, tolerance $\epsilon>0$, (optional) horizon $H$

\STATE \textbf{/* Make the model stationary if finite-horizon */} 
\IF{$H$ is specified}
    \STATE $S \leftarrow S\cup \{t\}$
    \FOR{$s \in S, a \in A_1$}
    \IF{$t \geq H$}
    \STATE $R(s,a)\leftarrow  0$ $\quad \forall s \in S, a \in A_1$\\
    \STATE $p(s|s,a)\leftarrow1$\\
    \STATE $p(s' |s, a)\leftarrow 0, \forall s' \neq s$
    \ENDIF
    \ENDFOR
    \COMMENT{add time $t$ to state; add absorbing terminal dynamics after $t=H$}
\ENDIF

\STATE \textbf{/* Reduce SA-MDP to a constrained zero-sum OS-POSG */} 
\STATE $\widehat G \leftarrow \textsc{ParallelizeToOSPOSG}(M)$
\COMMENT{dummy states/actions; $\widehat\gamma=\sqrt{\gamma}$; attacker action encodes reported observation}

\STATE $\widehat G \leftarrow \textsc{ImposeStateDependentConstraints}(\widehat G,B)$
\COMMENT{set $L_1(s)\gets\{\bot_1\}$ and $L_2(s)\gets B(s)$ on true states;set $L_1(\bar s)\gets A_1$ and $L_2(\bar s)\gets\{\bot_2\}$ on dummy states}

\STATE \textbf{/* Solve $\widehat G$ using HSVI (Hor\'ak et al.) */} 
\STATE $(V^{\Gamma}_{LB},V^{\Upsilon}_{UB},\Gamma,\Upsilon) \leftarrow \textsc{HSVI}(\widehat G,b^0,\epsilon)$

\STATE \textbf{/* Strategy extraction (online policies) */} 
\STATE $\sigma_1 \leftarrow \textsc{ContinualResolving}( \widehat G,\Gamma,b^0)$
\STATE $\sigma_2 \leftarrow \textsc{UpperBoundStagePolicy}( \widehat G, V^{\Upsilon}_{UB}, b^0)$

\STATE \textbf{/* Map back to SA-MDP by collapsing sub-steps */} 
\STATE $(\pi_1,\pi_2)\leftarrow \textsc{CollapseDummySteps}(\sigma_1,\sigma_2)$
\COMMENT{drop dummy turns; attacker action $\equiv$ perturbed observation}
 
\STATE \textbf{return} $(\pi_1,\pi_2,\,V^{\Gamma}_{LB}(b^0),\,V^{\Upsilon}_{UB}(b^0))$
\end{algorithmic}
\end{algorithm}
\begin{algorithm}[bth]
\caption{\textsc{HSVI} for constrained zero-sum OS-POSGs, blue indicates change relative to \citet{horak2023solving}}
\label{alg:constrained_hsvi}
\begin{algorithmic}[1]
\REQUIRE Constrained OS-POSG $G=(S,A_1,A_2,O,T,R,\gamma,L_1,L_2)$, initial belief $b^0$, tolerance $\epsilon>0$, neighborhood parameter $D$

\STATE Initialize $V^{\Gamma}_{LB}$ and $V^{\Upsilon}_{UB}$ (lower/upper bounds)
\WHILE{$\textsc{Excess}(b^0,0) > 0$}
    \STATE \textsc{Explore}$(b^0,0)$
\ENDWHILE
\STATE \textbf{return} $(V^{\Gamma}_{LB},V^{\Upsilon}_{UB},\Gamma,\Upsilon)$

\vspace{1mm}
\STATE \textbf{procedure} \textsc{Explore}$(b^t,t)$
\STATE $(\pi^{LB}_1,\pi^{LB}_2)\leftarrow \textsc{StageGameEquilibrium}(b^t,V^{\Gamma}_{LB})$
\COMMENT{\textcolor{blue}{$\pi^{LB}_1\in\Delta(L_1(h_1))$;
$\pi^{LB}_2(\cdot\mid s)\in\Delta(L_2(s))$ for $s\in S(h_1)$}}
\STATE $(\pi^{UB}_1,\pi^{UB}_2)\leftarrow \textsc{StageGameEquilibrium}(b^t,V^{\Upsilon}_{UB})$
\STATE \textsc{PointBasedUpdate}$(b^t, V^{\Gamma}_{LB}, V^{\Upsilon}_{UB})$
\STATE $(a_1^\star,o^\star)\leftarrow \textsc{ForwardExploreHeuristic}(b^t,t,\pi^{UB}_1,\pi^{LB}_2)$
\IF{$\Pr_{b^t,\pi^{UB}_1,\pi^{LB}_2}[a_1^\star,o^\star]\cdot \textsc{Excess}(\tau(b^t,a_1^\star,\pi^{LB}_2,o^\star),t{+}1) > 0$}
    \STATE \textsc{Explore}$(\tau(b^t,a_1^\star,\pi^{LB}_2,o^\star),t{+}1)$
\ENDIF
\STATE \textsc{PointBasedUpdate}$(b^t, V^{\Gamma}_{LB}, V^{\Upsilon}_{UB})$
\STATE \textbf{end procedure}
\end{algorithmic}

\end{algorithm}
\FloatBarrier
We use a linear programming routine to solve for optimal strategies and values at each sub-game node. The code is implemented in C++ with a Python interface and uses linear programming to solve the stage-game for HSVI and the Markovian equilibrium strategies. We use CPLEX with an educational license to solve linear programming routines. \citep{cplex-manual}

Several implementation-specific modifications are made to our HSVI algorithm to improve computational efficiency. We cache commonly explored state-actions pairs. We also first group histories that lead to the same sets of observations at time $t$ and then subsequently split information sets that overlap only due to the presence of a shared state in the feasibility set of the adversary into separate information sets by duplicating the state in our game tree. This enables solving many smaller convex problems versus solving fewer larger convex problems per stage and is equivalent to retaining separate extensive form (history-dependent) information sets associated with the game. However, computationally, this reduces the number of information sets that we need to keep track of and means that we only ``recover the information sets as needed'' at each stage.

Our performance is illustrated by  \cref{fig:scaling_performance} and \cref{tab:depth_scaling}.  In \cref{fig:scaling_performance}, we observe roughly linear scaling in number of states and log-linear scaling in the depth of the game. Min-max range are provided with shaded region, while bars are provided 1 std out. Normality of runs cannot necessarily be assumed, especially with only 5 seeds. Runs were declared to have converged when the excess-gap fell below an epsilon of 0.01. SA-MDP was assumed to have $|A_1|=4$, which was kept constant across tree-depth. An actual value of $\gamma=.99$ was used to approximate the $\gamma\sim 1$ case. The adversary's proximity set allows the adversary to perturb any current state to any adjacent state, with first and last states being adjacent. (So, at state 1, the adversary could perturb observations to \{4,1,2\}, etc.)
\begin{figure}[!htbp]
\centering
\caption{HSVI solve time versus tree-depth. }
\label{fig:scaling_performance}
\includegraphics[width=1\linewidth]
{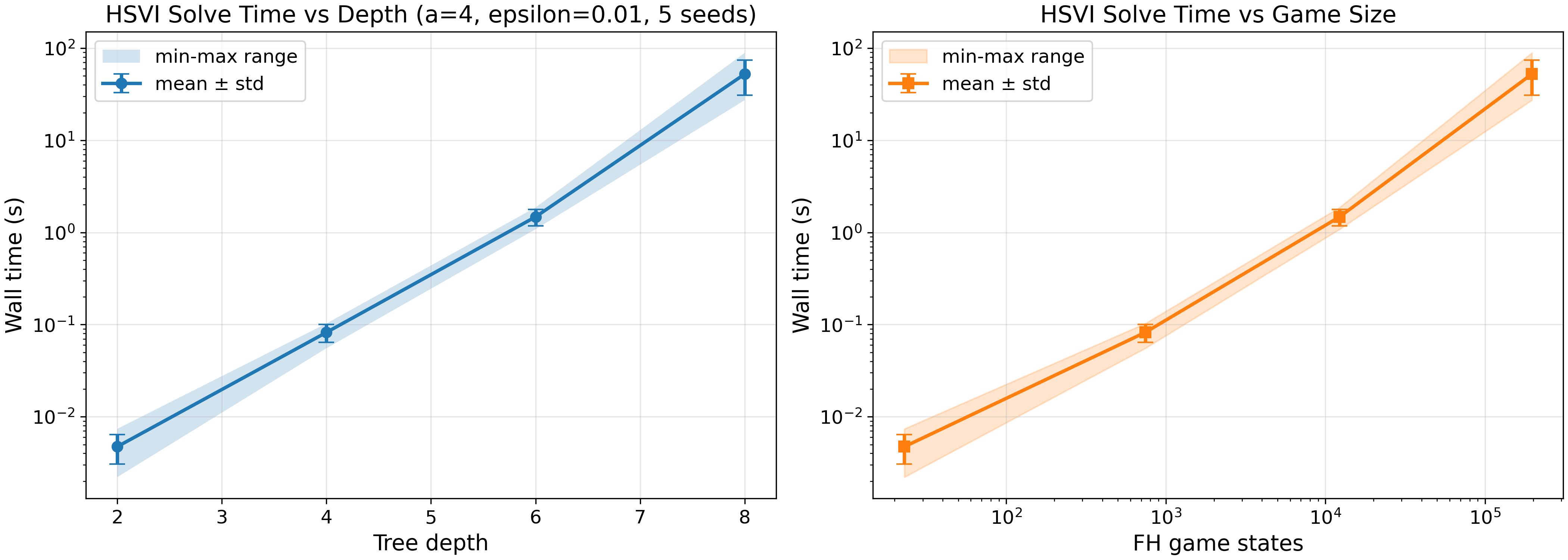}
\end{figure}

\FloatBarrier

\end{document}